\documentclass[10pt,a4paper,final]{article}
\pdfoutput=1
\usepackage{graphics,epsfig,graphicx,color}
\graphicspath{{../Figures/}}
\usepackage{subfig}
\usepackage{cite}
\usepackage{amssymb,latexsym}

\usepackage{setspace}

\usepackage{amsmath}

\usepackage{mathrsfs,amsfonts}

\usepackage{amsthm,amsxtra}

\usepackage{showkeys}
\usepackage{stmaryrd}

\usepackage{algorithm,algpseudocode}
\usepackage{mathtools}
\newif\ifPDF
\ifx\pdfoutput\undefined
\PDFfalse
\else
\ifnum\pdfoutput > 0
\PDFtrue
\else
\PDFfalse
\fi
\fi

\ifPDF
\usepackage[debug,pdftex,colorlinks=true, 
linkcolor=blue, bookmarksopen=false,
plainpages=false,pdfpagelabels]{hyperref}
\else
\usepackage[dvips]{hyperref}
\fi

\usepackage{fancyhdr}

\newtheorem{theorem}{Theorem}[section]
\newtheorem{lemma}[theorem]{Lemma}
\newtheorem{definition}[theorem]{Definition}

\newtheorem{remark}[theorem]{Remark}

\newcommand{\eps}{\varepsilon}

\newcommand{\sfc}{\mathsf c}

\newcommand{\bbR}{\mathbb R} \newcommand{\bbS}{\mathbb S}

 \newcommand{\bn}{\mathbf n}
 
\newcommand{\bq}{\mathbf q} 
\newcommand{\bs}{\mathbf s} \newcommand{\bt}{\mathbf t}
 \newcommand{\bv}{\mathbf v} 
 \newcommand{\bx}{\mathbf x} 
\newcommand{\by}{\mathbf y} \newcommand{\bz}{\mathbf z}

\newcommand{\cA}{\mathcal A} \newcommand{\cB}{\mathcal B}

 \newcommand{\cH}{\mathcal H}
\newcommand{\cI}{\mathcal I} \newcommand{\cJ}{\mathcal J}
\newcommand{\cK}{\mathcal K} \newcommand{\cL}{\mathcal L}
 
\newcommand{\cO}{\mathcal O}  
 \newcommand{\cR}{\mathcal R}
\newcommand{\cS}{\mathcal S} \newcommand{\cT}{\mathcal T}
\newcommand{\cU}{\mathcal U}

\newcommand{\fW}{\mathfrak W}
\newcommand{\opH}{\mathcal H}

\newcommand{\intS}[1]{\int_{\bbS^{d-1}} #1}

\DeclareMathOperator{\dist}{\mathsf d}

\DeclareMathOperator*{\argmin}{arg\,min}

\newenvironment{keywords}
{\noindent{\bf Key words.}\small}{\par\vspace{1ex}}
\newenvironment{AMS}
{\noindent{\bf AMS subject classifications 2000.}\small}{\par}

\makeatletter
\newcommand{\chapterauthor}[1]{%
	{\parindent0pt\vspace*{-25pt}%
		\linespread{1.1}\large\scshape#1%
		\par\nobreak\vspace*{35pt}}
	\@afterheading%
}

\newsavebox\myboxA
\newsavebox\myboxB
\newlength\mylenA

\newcommand*\xoverline[2][0.75]{%
	\sbox{\myboxA}{$\m@th#2$}%
	\setbox\myboxB\null% Phantom box
	\ht\myboxB=\ht\myboxA%
	\dp\myboxB=\dp\myboxA%
	\wd\myboxB=#1\wd\myboxA% Scale phantom
	\sbox\myboxB{$\m@th\overline{\copy\myboxB}$}%  Overlined phantom
	\setlength\mylenA{\the\wd\myboxA}%   calc width diff
	\addtolength\mylenA{-\the\wd\myboxB}%
	\ifdim\wd\myboxB<\wd\myboxA%
	\rlap{\hskip 0.5\mylenA\usebox\myboxB}{\usebox\myboxA}%
	\else
	\hskip -0.5\mylenA\rlap{\usebox\myboxA}{\hskip 0.5\mylenA\usebox\myboxB}%
	\fi}

\makeatother
\title{Inverse transport problem in fluorescence ultrasound modulated optical tomography with angularly averaged measurements}
\author{Wei Li, Yang Yang, Yimin Zhong}
\begin{document}
\maketitle
\begin{abstract}
	We consider an inverse transport problem in fluorescence ultrasound modulated optical tomography (fUMOT) with angularly averaged illuminations and measurements. We study the uniqueness and stability of the reconstruction of the absorption coefficient and the quantum efficiency of the fluorescent probes. Reconstruction algorithms are proposed and numerical validations are performed. This paper is an extension of~\cite{li2018hybrid}, where a diffusion model for this problem was considered.
\end{abstract} 

% REQUIRED

% REQUIRED
%%%%%%%%%%%%%%%%%%%%%%%%%%%%%%%%%%%%%%%%%%%%%%%%%%%%%%%%%%%%%%%
%%
%% Introduction
%%
%%%%%%%%%%%%%%%%%%%%%%%%%%%%%%%%%%%%%%%%%%%%%%%%%%%%%%%%%%%%%%%%
\section{Introduction}\label{SEC:1}
Fluorescence optical tomography (FOT) is a popular imaging modality for biomedical and preclinical research~\cite{arridge2009optical,chang1995fluorescence,chaudhari2005hyperspectral, corlu2007three,milstein2003fluorescence}.  Upon illumination by a laser pulse, fluorescent probes are exited to a metastable state and later decay to the ground state by emitting photons at a lower frequency. The emitted light and the residual excitation light are detected at the boundary for the reconstruction of the spatial concentration and lifetimes of the fluorophores. 

Fluorescence ultrasound modulated optical tomography (fUMOT) is a series of FOT experiments performed under varying acoustic modulation~\cite{yuan2009ultrasound,yuan2008mechanisms,yuan2009microbubble,liu2014ultrasound}. The acoustic modulation perturbs the optical properties of the tissue sample, allowing the measurements to provide internal information about the optical field. As  the fluorescent probes have high optical contrast and tissues are acoustically homogeneous, fUMOT is expected to provide stable high contrast reconstructions with resolution comparable to the wavelength of the acoustic field. The availability of the internal data and the wellposedness of the inverse problem is generic for hybrid imaging modalities \cite{bal2010inverse,marks1993comprehensive,kempe1997acousto,
 granot2001detection,wang1995continuous,ku2005deeply,jiao2002two, ammari2014UMOT}.

Light propagation in tissues obeys the radiative transport equation (RTE)~\cite{ishimaru1978wave}. When the tissue environment is highly scattering, the RTE can be approximated by the diffusion equation with a suitable boundary condition~\cite{ishimaru1978wave,bal2006radiative}. fUMOT in the diffusion regime has been studied in our previous work~\cite{li2018hybrid}. However, the diffusion approximation fails in the following two cases: when the tissue is optically thin, the characteristic length is at the same order as the transport mean free path, thus the boundary layer effect cannot be neglected; and when the scattering or the illumination source is highly anisotropic, the optical field is necessarily anisotropic {near the source}.
In an inverse transport problem, the illuminations and measurements at the boundary can be time dependent or time independent, and angularly resolved or angularly averaged~\cite{bal2009inverse,bal2008stability,bal2007inverse}. Time dependent measurements and angularly resolved measurements are mathematically preferable since they preserve more singularities and permit more stable and more resolved reconstruction. However, in practice, the photon transport process is too fast for accurate time dependent measurements, and angularly resolved measurements are too sensitive to noise due to possibly low particle counts in certain directions. { That is, in most practical applications, time independent and angularly averaged illuminations and measurements are less expensive and more reliable}~\cite{bal2007inverse,bal2009inverse,bal2008inverse}.

In this paper, we study fUMOT in the radiative transport regime with time independent and angularly averaged illumination and measurements. We derive the mathematical model for fUMOT in the transport regime following the works~\cite{li2018hybrid,ren2013quantitative}. Let $u(\bx ,\bv, t)$ and $w(\bx, \bv, t)$ be the excitation and emission photon densities at location $\bx\in \Omega$, along direction $\bv\in  \bbS^{d-1}$ at time $t\in\bbR^{+}$.  The governing equations of fluorescence optical tomography (FOT) are 
\begin{equation}\label{EQ:RTE}
\begin{split}
\frac{1}{c} \partial_t u(\bx, \bv, t) + \bv\cdot \nabla u(\bx, \bv, t) + (\sigma_{x,a}(\bx)  + \sigma_{x,f}(\bx) + \sigma_{x,s}(\bx)) u(\bx, \bv, t) \\= \sigma_{x,s}(\bx)\int_{\bbS^{d-1}} p(\bv\cdot \bv') u(\bx, \bv', t) d\bv'\quad &\text{in } X \times\bbR^{+}, \\
\frac{1}{c}\partial_t w(\bx, \bv, t) +  \bv\cdot \nabla w(\bx, \bv, t) + (\sigma_{m,a}(\bx) + \sigma_{m,f}(\bx) + \sigma_{m,s}(\bx) ) w(\bx, \bv, t)\\ = \sigma_{m,s}(\bx)\int_{\bbS^{d-1}} p(\bv\cdot\bv') w(\bx, \bv', t) d\bv' + S(\bx, t)\quad &\text{in } X\times \bbR^{+},\\
u(\bx, \bv, t) = g(\bx, \bv, t), \quad w(\bx, \bv, t) = 0\quad& \text{on } \Gamma_{-}\times\bbR^{+}, \\
u(\bx, \bv, t) = 0, \quad w(\bx, \bv, t) = 0\quad&\text{on } X\times \{0\}.
\end{split}
\end{equation}
Here, $\Omega\subset \bbR^d (d = 2, 3)$ is the domain of interest, $X = \Omega\times \bbS^{d-1}$ denotes the phase space, $\Gamma_{\pm} = \{(\bx, \bv)\in \partial \Omega\times \bbS^{d-1} | \pm \bn_{\bx}\cdot \bv > 0 \}$ are the incoming and outgoing boundary sets, $g(\bx, \bv, t) = g(\bx, \bv)\delta(t - 0^{+})$ is the external excitation laser pulse, and {we assume the reflection at the interface $\partial \Omega$ is negligible}.  
$\sigma_{x,a}$ (resp. $\sigma_{m,a}$) is the intrinsic absorption coefficient of the medium at the excitation wavelength (resp. emission wavelength), $\sigma_{x,s}$ (resp. $\sigma_{m,s}$) is the intrinsic scattering coefficient of the medium at the excitation wavelength (resp. emission wavelength), and $\sigma_{x,f}$ (resp. $\sigma_{m,f}$) is the \emph{absorption coefficient of the fluorophores} at the excitation wavelength (resp. emission wavelength). The emission source term $S(\bx, t)$ is proportional to the radiant energy and given by
\begin{equation}
S(\bx, t) = \eta(\bx) \sigma_{x,f} (\bx) \int_0^t \frac{1}{\tau} e^{-\frac{t-s}{\tau}} \left(\int_{\bbS^{d-1}} u(\bx, \bv, s) d\bv \right) ds,
\end{equation} 
where $\eta(\bx)$ is the \emph{quantum efficiency} or \emph{quantum yield} of the fluorophores and $\tau$ is the fluorescence lifetime of the excited state. The integral kernel $p(\bv\cdot \bv')$ is the scattering phase function, which gives the angular distribution of light intensity scattered by particle collision. With slight abuse of notation, we set $u(\bx, \bv) = \int_{0}^{\infty} u(\bx, \bv, t) dt$,  $w(\bx, \bv) = \int_0
^{\infty} w(\bx, \bv, t) dt$. Then we integrate the system~\eqref{EQ:RTE} over time. Noticing the fact that $u(\bx, \bv, \infty) = w(\bx, \bv, \infty) = 0$,  we obtain a stationary RTE system for these time-integrated quantities.
\begin{equation}\label{EQ:SRTE}
\begin{aligned}
\bv\cdot \nabla u(\bx, \bv) + (\sigma_{x,a}  + \sigma_{x,f} + \sigma_{x,s}) u(\bx, \bv) &= \sigma_{x,s}\int_{\bbS^{d-1}} p(\bv\cdot \bv') u(\bx, \bv') d\bv' &\text{ in }&X,\\
\bv\cdot \nabla w(\bx, \bv) + (\sigma_{m,a} + \sigma_{m,f} + \sigma_{m,s}) w(\bx, \bv) &=
\sigma_{m,s} \int_{\bbS^{d-1}} p(\bv\cdot \bv') w(\bx, \bv') d\bv' \\&\quad+ \eta \sigma_{x,f} \int_{\bbS^{d-1}} u(\bx, \bv) d\bv &\text{ in }&X,\\
u_{}(\bx, \bv) = g(\bx, \bv)&, \quad w_{}(\bx, \bv) = 0 &\text{on }&\Gamma_{-}.
\end{aligned}
\end{equation}
In practice, the coefficient $\sigma_{m,f}$ is extremely small compared to the other coefficients \cite[Fig. 1.7]{sauer2011handbook}, therefore we set it to zero hereafter. {For simplicity, in what follows, we consider the isotropic illuminations only, namely $g(\bx,\bv)=g(\bx)$.}

Similar to~\cite{li2018hybrid,chung2017inverse}, we consider the plane wave ultrasound modulation in the form of ${P}(\bx, t) = A\cos(\omega t) \cos(\bq \cdot \bx + \phi)$, where $A$ is the amplitude, $\omega$ is the frequency, $\bq$ is the wave vector and $\phi$ is the initial phase. Under the acoustic modulation, the optical coefficients take the form~\cite{bal2016ultrasound,bal2014ultrasound}
\begin{equation}\label{EQ:MODCOEF}
\begin{aligned}
\sigma_{x,s}^{\eps}(\bx) &= (1 +\eps\cos(\bq \cdot \bx + \phi))\sigma_{x,s}(\bx),\\
\sigma_{m,s}^{\eps}(\bx) &= (1 + \eps\cos(\bq \cdot \bx + \phi))\sigma_{m,s}(\bx), \\
\sigma_{x,a}^{\eps}(\bx) &= (1 + \eps\cos(\bq \cdot \bx + \phi) )\sigma_{x,a}(\bx),\\
\sigma_{m,a}^{\eps}(\bx)&=(1+\eps\cos(\bq \cdot \bx + \phi))\sigma_{m,a}(\bx),\\
\sigma_{x,f}^{\eps}(\bx) &=(1+\eps\cos(\bq \cdot \bx + \phi))\sigma_{x,f}(\bx),
\end{aligned}
\end{equation}
where $\eps=\frac{A \cos(\omega t)}{\rho c_s^2}\ll 1$, $\rho$ is the particle number density, and $c_s$ is the sound speed. {Note that the time variable $t$ in $\eps$ is the time on the acoustic time scale, which is approximately constant during the much faster optical process.} According to~\cite{bal2016ultrasound} the \emph{quantum efficiency} $\eta(\bx)$ is not modulated by the acoustic field. Combining this with the stationary RTE~\eqref{EQ:SRTE}, we obtain the governing equation for fUMOT in the transport regime,
\begin{equation}\label{EQ:MODSRTE}
\begin{aligned}
\bv\cdot \nabla u_{\eps}(\bx, \bv) + (\sigma^{\eps}_{x,a}  + \sigma^{\eps}_{x,f} + \sigma^{\eps}_{x,s}) u_{\eps}(\bx, \bv) &= \sigma^{\eps}_{x,s}\cK u_{\eps}(\bx, \bv) &\text{ in }&X,\\
\bv\cdot \nabla w_{\eps}(\bx, \bv) + (\sigma^{\eps}_{m,a}  + \sigma^{\eps}_{m,s}) w_{\eps}(\bx, \bv) &=
\sigma^{\eps}_{m,s}\cK w_{\eps}(\bx, \bv) + \eta \sigma^{{\eps}}_{x,f} \cI u_{\eps}(\bx)  &\text{ in }&X,\\
{{u_{\eps}(\bx, \bv) = g(\bx)}}&, \quad w_{\eps}(\bx, \bv) = 0 &\text{on }&\Gamma_{-}.
\end{aligned}
\end{equation}
where the integral operators $\cK$ and $\cI$ are defined as
\begin{equation}
\begin{aligned}
\cK f(\bx, \bv) &= \int_{\bbS^{d-1}} p(\bv\cdot \bv') f(\bx, \bv') d\bv', \quad \cI f(\bx) &= \int_{\bbS^{d-1}} f(\bx, \bv) d\bv.
\end{aligned}
\end{equation}
For the measurements, we record the angularly averaged boundary photon currents at both the excitation and the emission wavelengths~\cite{arridge2009optical,ren2010recent},
\begin{equation}\label{EQ:MEASURE}
\cJ u_{\eps} = \int_{\bbS^{d-1}} u_{\eps}(\bx, \bv) \bv\cdot \bn_{\bx} d\bv,\quad \cJ w_{\eps} = \int_{\bbS^{d-1}} w_{\eps}(\bx, \bv) \bv\cdot \bn_{\bx} d\bv.
\end{equation}
For a fixed external excitation source $g$, such boundary photon currents can be measured with multiple acoustic fields with various wave vectors $\bq$ and initial phases $\phi$. Therefore the measurement operator is
\begin{equation}
\Lambda^{\eps}(\bq, \phi) = (\cJ u_{\eps}, \cJ w_{\eps}) \Big|_{\partial \Omega}.
\end{equation}
The objective is to reconstruct the absorption coefficient of the fluorophores $\sigma_{x,f}(\bx)$ and the quantum efficiency $\eta(\bx)$ from the measurement operator $\Lambda^{\eps}$, assuming that the unperturbed background coefficients $\sigma_{x,a}$, $\sigma_{m,a}$, $\sigma_{x,s}$ and $\sigma_{m,s}$ have been reconstructed through other imaging methods~\cite{bal2010inverse,bal2011multi,ren2015inverse,ren2013quantitative}.

Due to the weak coupling between $\sigma_{x,f}$ and $\eta$ in the system~\eqref{EQ:MODSRTE}, there exists a two-step approach to simultaneously reconstructing
these two coefficients. Firstly, a nonlinear inverse medium problem at the excitation wavelength is solved to recover the absorption coefficient $\sigma_{x,f}$ using internal data derived from the excitation component of the measurement operator \eqref{EQ:MEASURE}. Secondly, with the knowledge of $\sigma_{x,f}$, we solve a linear inverse source problem
at the emission wavelength to find the quantum efficiency $\eta$  using internal data derived from the emission component of the measurement operator \eqref{EQ:MEASURE}. 

For the excitation stage, we consider two scenarios: (i) For the linearized problem with some smallness assumptions, we establish existence, uniqueness, and stability results with standard transport theory; (ii) For the nonlinear problem under the assumption that $\sigma_{x,f}$ is $\alpha$-H\"{o}lder continuous and known near the boundary, we propose a proximal reconstruction method for $\sigma_{x,f}$ which algebraically depends on
%a constructive reconstruction of $\sigma_{x,f}$, which algebraically depends on 
the internal data. The error of this reconstruction can be made arbitrarily small with a proper choice of the source, and the stability is of Lipschitz type. The key idea is to use an isotropic source that is localized around a set of points on the boundary. Under this illumination, on a line connecting two bright points on the boundary,  only the ballistic part of $u(\bx, \bv)$ contributes to the leading order term of the internal data, whereas the scattering parts of $u(\bx, \bv)$ yield lower order terms. It is an analogue of the highly collimated source function in~\cite{chung2017inverse}, where angularly resolved illuminations and measurements are allowed. 

At last, we make a few comments on some relevant inverse transport problems. In the absence of acoustic modulation, inverse problems for the time independent RTE with angularly averaged measurements and illuminations are mostly open~\cite{bal2009inverse, zhao2018instability}.
%Without acoustic modulation, inverse problems with the time independent RTE with angularly averaged measurements is mostly open. 
The equation considered in this setting is the first equation in \eqref{EQ:SRTE}, where the sum $\sigma_{x,a} + \sigma_{x,s} + \sigma_{x,f}$ is denoted by $\sigma_{x,tf}$~\cite{bal2009inverse}. When $\sigma_{x,tf}$ is unknown or $\sigma_{x,s}$ is unknown, there is no uniqueness result for the reconstruction of $\sigma_{x,tf}$ or $\sigma_{x,s}$.
When only $\sigma_{x,s}$ is unknown and $\sigma_{x,tf}$ and $\sigma_{x,s}$ are small, recovering $\sigma_{x,s}$ is severely unstable~\cite{bal2008inverse}. 
%{\color{red}{One seemingly time independent RTE whose angularly averaged measurements preserve much singularity is the stationary RTE with time-harmonic source of a large frequency~\cite{bal2007harmonic}. However, a large frequency corresponds intrinsically to good temporal sampling thus suffers the same limitations as time dependent illuminations and measurements.}} 
In the presence of acoustic modulation, inverse problems with the time independent RTE with angularly resolved measurements are studied in~\cite{chung2017inverse,bal2016ultrasound}.

%When angularly resolved measurements are allowed, the inverse problem can be solved in a constructive way using a highly collimated source function $g(\bx, \bv)$ by the Schwartz kernel decomposition of the Albedo operator~\cite{chung2017inverse}. 

The rest of the paper is organized as follows. In Section~\ref{SEC:INTERNAL}, we extract some internal data from the measurements. A few general properties of the inverse problem are established in Section~\ref{SEC:PROP}. In Section~\ref{SEC:SIGMA}, we reconstruct $\sigma_{x,f}$ from the internal data at the excitation stage. We give results on the uniqueness and stability of $\sigma_{x,f}$ for the linearized problem, and provide an algebraic reconstruction formula for the nonlinear problem. In Section~\ref{SEC:ETA}, assuming $\sigma_{x,f}$ has been successfully reconstructed, we recover $\eta$ from the internal data at the emission stage. The numerical experiments on synthetic data are presented in Section~\ref{SEC:NUM} for validation.

%%%%%%%%%%%%%%%%%%%%%%%%%%%%%%%%%%%%%%%%%%%%%%%%%%%%%%%%%%%%%
%%
%%  Internal data
%%
%%%%%%%%%%%%%%%%%%%%%%%%%%%%%%%%%%%%%%%%%%%%%%%%%%%%%%%%%%%%%
\section{Internal data}\label{SEC:INTERNAL}
In analogy to~\cite{bal2016ultrasound}, we introduce the self-adjoint operators $A_{\eps}$ and $A_0$ defined by 
\begin{equation}
\begin{aligned}
A_{\eps} f = - (\sigma^{\eps}_{x,a} + \sigma^{\eps}_{x,f} + \sigma^{\eps}_{x,s}) f  + \sigma^{\eps}_{x,s} \cK f,\\
A_{0} f = - (\sigma_{x,a} + \sigma_{x,f} + \sigma_{x, s}) f  + \sigma_{x,s} \cK f;
\end{aligned}
\end{equation}
then the modulated solution $u_{\eps}$ satisfies
\begin{equation}\label{eq:A0}
(\bv\cdot \nabla -A_{\eps}) u_{\eps}(\bx, \bv) = 0.
\end{equation}
{ We then consider the auxiliary function $\mathfrak{U}(\bx, \bv) :=  u(\bx, - \bv)$, which satisfies
the adjoint radiative transfer equation
\begin{equation}\label{eq:A1}
(-\bv\cdot \nabla - A_0) \mathfrak{U}(\bx, \bv) = 0
\end{equation}
with $\mathfrak{U}(\bx, \bv) = g(\bx)$ on the outgoing boundary $\Gamma_{+}$.} The quantity
\begin{equation}
\cJ \mathfrak{U}(\bx) := \int_{\bbS^{d-1}} \bv \cdot \bn_{\bx} \mathfrak{U}(\bx, \bv) d\bv = \int_{\bbS^{d-1}} \bv \cdot \bn_{\bx} u(\bx, -\bv) d\bv = - \int_{\bbS^{d-1}} \bv \cdot \bn_{\bx} u(\bx,\bv) d\bv = - \cJ u(\bx) 
\end{equation}
is available from the measurements. 
Computing $(A_{\eps} - A_0) f$ with the modulated coefficients \eqref{EQ:MODCOEF}, we find that
\begin{equation}
(A_{\eps} - A_0)f = \eps \cos(\bx\cdot \bq + \phi) \left(-( \sigma_{x,a} +  \sigma_{x,f} +  \sigma_{x,s})f + \sigma_{x,s}\cK f\right).
\end{equation}
Multiplying the equations~\eqref{eq:A0} and~\eqref{eq:A1} by $\mathfrak{U}_0$ and $u_{\eps}$ respectively, we obtain
\begin{equation}
\int_{\bbS^{d-1}}\int_{\Omega} \left((A_{\eps} - A_0) u_{\eps} \right) \mathfrak{U}_0 d\bx d\bv = \int_{\bbS^{d-1}} \int_{\Omega} \bv\cdot \nabla (u_{\eps} \mathfrak{U}) d\bx d\bv= \int_{\bbS^{d-1}} \int_{\partial \Omega} \bn_{\bx}\cdot \bv u_{\eps} \mathfrak{U}  d\bs d\bv.
\end{equation}
Since the boundary illumination is isotropic, the right-hand side is equal to
\begin{equation}
\int_{\Gamma_{-}} \bn_{\bx}\cdot \bv g(\bx) \mathfrak{U} d\bs d\bv + \int_{\Gamma_{+}}  \bn_{\bx}\cdot \bv u_{\eps} g(\bx) d\bs d\bv = \int_{\partial\Omega} \left( \cJ \mathfrak{U} + \cJ u_{\eps}\right)g(\bx)  d\bs(\bx).
\end{equation}
The right-hand side is known from the measurements by noticing that $ \mathfrak{U}(\bx, \bv) = u(\bx, - \bv)$. When $\eps$ is sufficiently small, {{we write the solution $u_{\eps}$ in an asymptotic expansion
\begin{equation}\label{EQ:ASY}
u_{\eps} = u + \eps u_1 + \eps^2 u_2 + \cdots.
\end{equation}
}}Then the following quantity is known up to higher order terms in $\eps$,
\begin{equation}
\int_{\bbS^{d-1}} \int_{\Omega} \cos(\bx\cdot \bq + \phi) \left(-( \sigma_{x,a} +  \sigma_{x,f} +  \sigma_{x,s})u + \sigma_{x,s}\cK u \right) \mathfrak{U} d\bx d\bv.
\end{equation}
Varying $\bq$ and $\phi$ and performing the inverse Fourier transform, we obtain the internal data $H(\bx)$ for the excitation stage
\begin{equation}\label{EQ:H}
\begin{aligned}
H(\bx) &= \int_{\bbS^{d-1}} \left(-( \sigma_{x,a} + \sigma_{x,f} + \sigma_{x,s})u +  \sigma_{x,s}\cK u \right)\mathfrak{U} d\bv \\
&=  - \sigma_{x,tf}   \int_{\bbS^{d-1}} u(\bx, \bv)\mathfrak{U}(\bx, \bv) d\bv  +  \sigma_{x,s} \int_{\bbS^{d-1}} \cK u(\bx, \bv) \mathfrak{U}(\bx, \bv) d\bv \\
&=  - \sigma_{x,tf}   \int_{\bbS^{d-1}} u(\bx, \bv) u(\bx, -\bv) d\bv  +  \sigma_{x,s} \int_{\bbS^{d-1}} \cK u(\bx, \bv) u(\bx, -\bv) d\bv,
\end{aligned}
\end{equation}
 { where $\sigma_{x,tf} := \sigma_{x,t} + \sigma_{x,f}$ and $\sigma_{x,t} := \sigma_{x,a} + \sigma_{x,s}$ denote the total absorption coefficient at the excitation wavelength with and without the fluorescence.} Similarly, to compute the internal data at the emission stage, we define auxiliary functions $\mathfrak{W}$ and $\varphi$ by the equations
\begin{equation}\label{EQ:W}
\begin{aligned}
-\bv\cdot \nabla \mathfrak{W}(\bx, \bv) + (\sigma_{m,a} + \sigma_{m,s})\mathfrak{W}(\bx, \bv) &= \sigma_{m,s} \cK\mathfrak{W}(\bx, \bv)  &\text{ in }&X,\\
\mathfrak{W}(\bx, \bv) &= h(\bx) &\text{ on }&\Gamma_{+},
\end{aligned}
\end{equation}
for some strictly positive function $h(\bx)\in L^{\infty}(\partial\Omega)$,  and 
\begin{equation}\label{EQ:vphi}
\begin{aligned}
-\bv\cdot \nabla \varphi(\bx, \bv) + \sigma_{x,tf} \varphi(\bx, \bv) &= \sigma_{x,s} \cK\varphi(\bx, \bv) + \eta \sigma_{x,f}\cI \mathfrak{W}(\bx) &\text{ in }&X,\\
\varphi(\bx, \bv) &= 0 &\text{ on }& \Gamma_{+}. 
\end{aligned}
\end{equation}
{
Multipling \eqref{EQ:W} by $w_{\eps}$ and \eqref{EQ:MODSRTE} by $\mathfrak{W}$, we obtain 
\begin{equation}\label{EQ:etaint1}
\begin{aligned}
&\quad \eps \int_{\Omega} \cos(\bx\cdot \bq + \phi)  \left(-(\sigma_{m,a} + \sigma_{m,s})\int_{\bbS^{d-1}} w(\bx, \bv) \mathfrak{W}(\bx, \bv) d\bv + \sigma_{m,s}\int_{\bbS^{d-1}} \cK w(\bx, \bv) \mathfrak{W}(\bx, \bv) d\bv \right)d\bx\\
&=\int_{\bbS^{d-1}} \int_{\Omega} \bv\cdot \nabla (w_{\eps} \mathfrak{W}) d\bx d\bv   -   \int_{\Omega} \eta \sigma^{\eps}_{x,f}(\cI u)(\cI \mathfrak{W}) d\bx\\
&= \int_{\bbS^{d-1}} \int_{\partial \Omega} \bn_{\bx}\cdot \bv w_{\eps} \mathfrak{W}  d\bs d\bv  -   \int_{\Omega} \eta \sigma^{\eps}_{x,f}(\cI u)(\cI \mathfrak{W}) d\bx.
\end{aligned}
\end{equation}
Similarly, for \eqref{EQ:vphi} and \eqref{EQ:MODSRTE}, we obtain up to higher orders in $\eps$,
\begin{equation}\label{EQ:etaint2}
\begin{aligned}
&\quad \eps \int_{\Omega} \cos(\bx\cdot \bq + \phi)  \left( -\sigma_{x,tf}\int_{\bbS^{d-1}} u(\bx, \bv) \varphi(\bx, \bv) d\bv + \sigma_{x,s}\int_{\bbS^{d-1}} \cK u(\bx, \bv) \varphi(\bx, \bv) d\bv \right)d\bx\\
&= \int_{\bbS^{d-1}} \int_{\Omega} \bv\cdot \nabla (u_{\eps}  \varphi) d\bx d\bv + \int_{\Omega} \eta \sigma_{x,f}(\cI u)(\cI \mathfrak{W}) d\bx\\
&= \int_{\bbS^{d-1}} \int_{\partial \Omega} \bn_{\bx}\cdot \bv u_{\eps}  \varphi d\bs d\bv + \int_{\Omega} \eta \sigma_{x,f}(\cI u)(\cI \mathfrak{W}) d\bx\\
& \approx \int_{\bbS^{d-1}} \int_{\partial \Omega} \bn_{\bx}\cdot \bv u \varphi d\bs d\bv + \eps \int_{\bbS^{d-1}} \int_{\partial \Omega} \bn_{\bx}\cdot \bv u_1  \varphi d\bs d\bv + \int_{\Omega} \eta \sigma_{x,f}(\cI u)(\cI \mathfrak{W}) d\bx.
\end{aligned}
\end{equation}
The sum of \eqref{EQ:etaint1} and \eqref{EQ:etaint2} gives 
\begin{equation}\label{EQ:etaint3}
\begin{aligned}
 &\quad \int_{\bbS^{d-1}} \int_{\partial \Omega} \bn_{\bx}\cdot \bv w_{\eps} \mathfrak{W}  d\bs d\bv  +  \int_{\bbS^{d-1}} \int_{\partial \Omega} \bn_{\bx}\cdot \bv u \varphi d\bs d\bv + \eps \int_{\bbS^{d-1}} \int_{\partial \Omega} \bn_{\bx}\cdot \bv u_1  \varphi d\bs d\bv\\
&\approx \eps \int_{\Omega} \cos(\bx\cdot \bq + \phi)  \Big(-(\sigma_{m,a} + \sigma_{m,s})\int_{\bbS^{d-1}} w(\bx, \bv) \mathfrak{W}(\bx, \bv) d\bv + \sigma_{m,s}\int_{\bbS^{d-1}} \cK w(\bx, \bv) \mathfrak{W}(\bx, \bv) d\bv \\
&\quad+\eta \sigma_{x,f}(\cI u)(\cI \mathfrak{W}) -\sigma_{x,tf}\int_{\bbS^{d-1}} u(\bx, \bv) \varphi(\bx, \bv) d\bv + \sigma_{x,s}\int_{\bbS^{d-1}} \cK u(\bx, \bv) \varphi(\bx, \bv) d\bv
 \Big)d\bx.
\end{aligned}
\end{equation}
The first term on left-hand side in \eqref{EQ:etaint3} is known from the measurements because
\begin{equation}
 \int_{\bbS^{d-1}} \int_{\partial \Omega} \bn_{\bx}\cdot \bv w_{\eps} \mathfrak{W}  d\bs d\bv = \int_{\Gamma_{+}} \bn_{\bx}\cdot \bv w_{\eps} h(\bx) d\bs d\bv =  \int_{\partial \Omega} \cJ w_{\eps}  h(\bx) d\bs.
\end{equation}
The second term on left-hand side is known from the boundary conditions. The third term is bounded by the Cauchy-Schwartz inequality
\begin{equation}
\left|\int_{\bbS^{d-1}} \int_{\partial \Omega} \bn_{\bx}\cdot \bv u_{1}  \varphi d\bs d\bv\right|\le \left(\int_{\bbS^{d-1}} \int_{\partial \Omega} |\bn_{\bx}\cdot \bv| |u_{1}|^2 d\bs\bv \right)^{1/2} \left(\int_{\bbS^{d-1}} \int_{\partial \Omega} |\bn_{\bx}\cdot \bv| |\varphi|^2 d\bs\bv \right)^{1/2},
\end{equation}
and from Lemma 2.2 in~\cite{agoshkov2012boundary},
\begin{equation}
\begin{aligned}
\int_{\bbS^{d-1}} \int_{\partial \Omega} |\bn_{\bx}\cdot \bv| |\varphi|^2 d\bs\bv = \int_{\Gamma_{-}} |\bn_{\bx}\cdot \bv| |\varphi|^2 d\bs\bv &\le c\left( \int_{\Gamma_{+}} |\bn_{\bx}\cdot \bv| |\varphi|^2 d\bs\bv + \|\eta\sigma_{x,f}\cI\mathfrak{W}\|^2_{L^2(\Omega)} \right)\\
&\le c \|\eta\sigma_{x,f}\cI\mathfrak{W}\|^2_{L^2(\Omega)},
\end{aligned}
\end{equation}
where the constant $c$ depends on $\Omega$ only. Experimentally, $\eta$ and $\sigma_{x,f}$ are usually spatially localized functions concentrated on the target cells such that $\|\eta\sigma_{x,f}\cI\mathfrak{W}\|_{L^2(\Omega)}\ll 1$, hence we omit this term from~\eqref{EQ:etaint3}.
}
Therefore, the internal data $S$ at the emission stage is
\begin{equation}\label{EQ:S}
\begin{aligned}
S(\bx) &= -(\sigma_{m,a} + \sigma_{m,s})\int_{\bbS^{d-1}} w(\bx, \bv) \mathfrak{W}(\bx, \bv) d\bv + \sigma_{m,s}\int_{\bbS^{d-1}} \cK w(\bx, \bv) \mathfrak{W}(\bx, \bv) d\bv \\
&\quad+\eta \sigma_{x,f}(\cI u)(\cI \mathfrak{W}) -\sigma_{x,tf}\int_{\bbS^{d-1}} u(\bx, \bv) \varphi(\bx, \bv) d\bv + \sigma_{x,s}\int_{\bbS^{d-1}} \cK u(\bx, \bv) \varphi(\bx, \bv) d\bv.
\end{aligned}
\end{equation}

{Under the assumption that  $\eps$ is sufficiently small, the internal data $H(\bx)$ and $S(\bx)$ given by \eqref{EQ:H} and \eqref{EQ:S} for all $\bx\in \Omega$ are available.}
In the diffusion regime, it is easy to check the above internal data $H$ and $S$ will be simplified to the internal data in~\cite{li2018hybrid}.
In the following, we will recover the unknown coefficients $(\sigma_{x, f}, \eta)$ from the internal data $(H,S)$ simultaneously. Since the coupling between $\sigma_{x,f}$ and $\eta$ is weak, we take a two-step reconstruction process, i.e., first reconstruct $\sigma_{x,f}$ from $H$ and then use the recovered coefficient to reconstruct the quantum efficiency $\eta$ as in~\cite{li2018hybrid}.

%%%%%%%%%%%%%%%%%%%%%%%%%%%%%%%%%%%%%%%%%%%%%%%%%%%%%%%%%%%
%%
%% General properties 
%%
%%%%%%%%%%%%%%%%%%%%%%%%%%%%%%%%%%%%%%%%%%%%%%%%%%%%%%%%%%%
\section{General properties of the inverse problems}\label{SEC:PROP}
In this section, we derive some general properties of the inverse problems of reconstructing $\sigma_{x,f}$ from the internal data $H$ and reconstructing $\eta$ from $S$ in the transport equations~\eqref{EQ:SRTE}. For any $1\leq p \leq \infty$, let  $L^p(X)$ (resp. $L^p(\Omega)$) denote the Lebesgue space of real-valued functions whose $p$-th power are Lebesgue integrable over $X$ (resp. $\Omega$), and $\cH_p^1(X)$ the space of $L^p(X)$ functions whose directional derivative along $\bv$ belongs to $L^p(X)$ as well, i.e.,  $\cH_p^1(X) := \{ f(\bx, \bv) : f\in L^p(X) \text{ and } \bv\cdot \nabla f \in L^p(X) \}$. We also let $L^p(\Gamma_{-})$ denote the space of functions that are the traces of $\cH_p^1(X)$ functions on $\Gamma_{-}$ under the norm $\|f\|_{L^p(\Gamma_{-})} := (\int_{\Gamma_{-}} |\bn(\bx) \cdot \bv| |f|^p d\bv d\bs )^{1/p}$, where $d\bs$ is the surface measure on $\partial\Omega$.  We make the flowing assumptions
\begin{enumerate}
	\item[($\mathfrak{A}$1).] The domain $\Omega$ is \emph{convex}  and simply connected, and $\partial\Omega$ is $C^2$.
	\item[($\mathfrak{A}$2).] The optical coefficients $\sigma_{x,a},  \sigma_{x,s}, \sigma_{m,a}, \sigma_{m, s}$ are bounded by some constants $\sfc_1$ and $\sfc_2$, with
	\begin{equation}
	0 < 	\sfc_1 < \sigma_{x,a}, \sigma_{x,s}, \sigma_{m,a}, \sigma_{m, s} < \sfc_2 < \infty.
	\end{equation} 
	The unknown coefficients $\sigma_{x,f}, \eta$ belong to the admissible sets $\cA_{\sigma}$ and $\cA_{\eta}$, respectively, where
	\begin{equation}
	\begin{aligned}
	&\cA_{\sigma} := \{ \sigma_{x,f}: 0 < \sfc_3 \le \sigma_{x,f} \le \sfc_4 < \infty \}, \\
	&\cA_{\eta} := \{ \eta:  0 \le \sfc_5 \le \eta \le  \sfc_6 < 1  \},
	\end{aligned}
	\end{equation}
	for some constants $\sfc_3, \sfc_4, \sfc_5$ and $\sfc_6$. 
	\item[($\mathfrak{A}$3).] The source function $g(\bx)$ is strictly positive, that is, there exists a constant $\sfc_7$ such that $0 < \sfc_7 \le g(\bx)$ for $\bx \in\partial\Omega$.
	\item[($\mathfrak{A}$4).] The scattering phase function $p(\bv\cdot \bv')$ 	is strictly positive and uniformly bounded and satisfies
	\begin{equation}
	\int_{\bbS^{d-1} }p(\bv\cdot \bv') d\bv' = 1,\quad 0 < \sfc_8 < p(\bv\cdot \bv') < \sfc_9 <\infty 
	\end{equation}
	for some constants $\sfc_8$ and $\sfc_9$.
\end{enumerate}
The above assumptions permit unique solutions $u(\bx, \bv),w(\bx,\bv) \in \cH^1_p(X)$ to RTE~\eqref{EQ:SRTE} for any given function $g(\bx) \in L^p(\partial\Omega)$ from the standard transport theory in ~\cite{agoshkov2012boundary}. Therefore the internal data $H$ and $S$ are well-defined for any $g(\bx)\in L^p(\partial\Omega)$ that satisfies the above assumptions. In the following, we show that $H$ and $S$ continuously depend on the unknown coefficients $\sigma_{x,f}$ and $\eta$ respectively.

\begin{theorem}\label{THM:FRECHET H}
	For any $g(\bx)\in L^p(\partial\Omega)$, suppose the assumptions ($\mathfrak{A}$1-$\mathfrak{A}$4) hold, then the { operator $\opH: L^{\infty}(\Omega)\rightarrow L^{p/2}(\Omega)$, which maps $\sigma_{x,f}$ to the internal data $H$,}  
	is Fr\'echet differentiable at any $\sigma_{x,f}\in\cA_{\sigma}$ in the direction $\delta\sigma_{x,f}\in L^{\infty}(\Omega)$ such that $\sigma_{x,f} + \delta\sigma_{x,f}\in \cA_{\sigma}$. The derivative is given by 
	%at any $\sigma_{x,f}$ in direction $\delta\sigma_{x,f}\in L^{\infty}(\Omega)$ that $\sigma_{x,f}\in\cA_{\sigma}$ and $\sigma_{x,f} + \delta\sigma_{x,f}\in \cA_{\sigma}$. The derivative is given by 
	\begin{equation}\label{EQ:HP}
	\begin{aligned}
	{\opH}'[\sigma_{x,f}](\delta\sigma_{x,f}) &= -\delta\sigma_{x,f} \int_{\bbS^{d-1}} u(\bx, \bv) u(\bx, -\bv) d\bv  -2 \sigma_{x,tf} \int_{\bbS^{d-1}} v(\bx, \bv) u(\bx,-\bv)  d\bv \\&\quad + 2\sigma_{x,s} \int_{\bbS^{d-1}} \cK v(\bx, \bv) u(\bx, -\bv) d\bv,
	\end{aligned}
	\end{equation}
	where $v(\bx, \bv)$ satisfies
	\begin{equation}\label{EQ:V}
	\begin{aligned}
	\bv\cdot \nabla v(\bx, \bv) + \sigma_{x,tf} v(\bx, \bv) &= \sigma_{x,s}\cK v(\bx,\bv) - \delta\sigma_{x,f} u &\text{ in }& X,\\
	v(\bx, \bv) &= 0\quad &\text{ on }&\Gamma_{-}.
	\end{aligned}
	\end{equation}
\end{theorem}
\begin{proof}
	Let $\tilde{\sigma}_{x,f} = \sigma_{x,f} + \delta \sigma_{x,f}$, $\tilde{u}$ be the solution to the first equation in~\eqref{EQ:SRTE} with coefficient $\tilde{\sigma}_{x,f}$, and $\tilde{H}$ be the corresponding internal data. Then $u' := \tilde{u}- u$ solves the transport equation
	\begin{equation}\label{EQ:u'}
	\begin{aligned}
	\bv \cdot \nabla u'(\bx, \bv) + \sigma_{x, tf} u'(\bx, \bv) &= \sigma_{x,s} \cK u'(\bx, \bv) - \delta\sigma_{x,f} \tilde{u}&\text{ in }& X,\\
	u'(\bx, \bv) &= 0\quad &\text{ on }&\Gamma_{-}.
	\end{aligned}
	\end{equation}
	Denote the difference between $v$ and the true perturbation $u'$ by $u'' := u' - v$.  We have that $u''$ satisfies the transport equation
	\begin{equation}\label{EQ:u''}
	\begin{aligned}
	\bv \cdot \nabla u''(\bx, \bv) + \sigma_{x, tf} u''(\bx, \bv) &= \sigma_{x,s} \cK u''(\bx, \bv) - \delta\sigma_{x,f} u'&\text{ in }& X,\\
	u''(\bx, \bv) &= 0\quad &\text{ on }&\Gamma_{-}.
	\end{aligned}
	\end{equation}	
	We now show that $\|u''\|_{L^p(X)}$ is of order $\|\delta\sigma_{x,f}\|^2_{L^{\infty}(\Omega)}$ using the standard theory of transport equations~\cite{agoshkov2012boundary}.
	The source term $\delta\sigma_{x,f}u'$ in \eqref{EQ:u'} is in $L^p(X)$, therefore $u'\in \cH_p^1(X)$ and there exist constants $\mathfrak{c}_1$ and $\mathfrak{c}_2$ such that
	\begin{equation}
	\|u'\|_{L^p(X)} \le \mathfrak{c}_1 \|\delta \sigma_{x,f} \tilde{u}\|_{L^p(X)}\le\mathfrak{c}_2\|\delta\sigma_{x,f}\|_{L^{\infty}(\Omega)} \|g\|_{L^p(\partial\Omega)}.
	\end{equation}
	It follows that the source term $\delta\sigma_{x,f}u'$ in \eqref{EQ:u''} lies in $L^p(X)$,  thus 
	\begin{equation}
	\|u''\|_{L^p(X)} \le \mathfrak{c}_1 \|\delta \sigma_{x,f} u'\|_{L^p(X)}\le \mathfrak{c}_1  \mathfrak{c}_2\|\delta\sigma_{x,f}\|^2_{L^{\infty}(\Omega)} \|g\|_{L^p(\partial\Omega)}.
	\end{equation}
%	When the coefficients satisfy the assumptions ($\mathfrak{A}$1-$\mathfrak{A}$4), the unique solutions $u,\tilde{u}\in \cH_p^1(X)$, hence $u'\in\cH_p^1(X)$, since $\delta\sigma_{x,f}u' \in L^p(X)$, then $u''\in \cH_p^1(X)$ as well. From the standard theory of transport equation~\cite{agoshkov2012boundary}, there exist constants $\mathfrak{c}_1$ and $\mathfrak{c}_2$ such that
%	\begin{equation}
%	\begin{aligned}
%	&\|u'\|_{L^p(X)} \le \mathfrak{c}_1 \|\delta \sigma_{x,f} \tilde{u}\|_{L^p(X)}\le\mathfrak{c}_2\|\delta\sigma_{x,f}\|_{L^{\infty}(\Omega)} \|g\|_{L^p(\partial\Omega)},\\
%	&\|u''\|_{L^p(X)} \le \mathfrak{c}_1 \|\delta \sigma_{x,f} u'\|_{L^p(X)}\le \mathfrak{c}_1  \mathfrak{c}_2\|\delta\sigma_{x,f}\|^2_{L^{\infty}(\Omega)} \|g\|_{L^p(\partial\Omega)}.
%	\end{aligned}
%	\end{equation}
	Hence $u$ is Fr\'echet differentiable with respect to $\sigma_{x,f}$ as a map from $L^{\infty}(\Omega)$ to $L^p(X)$. By the product rule, the Fr\'echet derivative of ${\opH}$ with respect to $\sigma_{x,f}$ is
	\begin{equation}
	\begin{aligned}
	{\opH}'[\sigma_{x,f}](\delta\sigma_{x,f}) &= -\delta\sigma_{x,f} \int_{\bbS^{d-1}} u(\bx, \bv) u(\bx, -\bv) d\bv  -2 \sigma_{x,tf} \int_{\bbS^{d-1}} v(\bx, \bv) u(\bx,-\bv)  d\bv \\&\quad + 2\sigma_{x,s} \int_{\bbS^{d-1}} \cK v(\bx, \bv) u(\bx, -\bv) d\bv.
	\end{aligned}
	\end{equation}
\end{proof}

\begin{theorem}\label{THM:FREDHOLM H}
	For any $g(\bx)\in L^{\infty}(\partial\Omega)$, suppose the assumptions ($\mathfrak{A}$1-$\mathfrak{A}$4) hold.  Then the Fr\'echet derivative ${\opH}'[\sigma_{x,f}]:L^2(\Omega)\rightarrow L^2(\Omega)$ is Fredholm.
\end{theorem}
\begin{proof}
	From the assumptions ($\mathfrak{A}$1-$\mathfrak{A}$4), the solution $u(\bx, \bv)$ is strictly positive, thus there exists a constant $\hat{\sfc} > 0$ such that
	 $$\int_{\bbS^{d-1}} u(\bx, \bv) u(\bx, -\bv) d\bv  > \hat{\sfc},\quad  \forall \bx\in \Omega.$$
	On the other hand, since $v(\bx, \bv) u(\bx, -\bv) \in \cH_2^1(X)$ and $\cK v(\bx, \bv) u(\bx, -\bv) \in \cH_2^1(X)$,
	by the averaging lemma~\cite{devore2001averaging,diperna1991lp}, we obtain 
	\begin{equation}
	\int_{\bbS^{d-1}} v(\bx, \bv) u(\bx,-\bv)  d\bv\in W^{2,1/2}(\Omega),\quad \int_{\bbS^{d-1}} \cK v(\bx, \bv) u(\bx, -\bv) d\bv \in W^{2,1/2}(\Omega).
	\end{equation}
	Then by the fact that the embedding from $W^{2, 1/2}(\Omega)$ to $L^2(\Omega)$ is compact, we obtain that ${\opH}'[\sigma_{x,f}]: L^2(\Omega)\rightarrow L^2(\Omega)$ is Fredholm.
\end{proof}
\begin{theorem}\label{THM:FREDHOLM S}
	For any $g(\bx)\in L^{\infty}(\Omega)$, suppose the assumptions ($\mathfrak{A}$1-$\mathfrak{A}$4) hold and $\sigma_{x,f}$ is known. {Then the linear operator $\cS : L^{2}(\Omega)\rightarrow L^{2}(\Omega)$, which maps $\eta$ to the internal data $S$, is Fredholm.}
	 
\end{theorem}
\begin{proof}
	Since $\sigma_{x,f}$ is known, $w(\bx, \bv)$ and $\varphi(\bx, \bv)$ are linear in $\eta$, hence ${\cS}$ is a linear functional of $\eta$. Since the auxiliary function $h(\bx)$ in~\eqref{EQ:W} is strictly positive, $\cI\mathfrak{W}$ is strictly positive over $\Omega$. Thus $\sigma_{x,f}(\cI u)(\cI \mathfrak{W})$ is strictly positive. On the other hand, since $w(\bx, \bv)\mathfrak{W}(\bx, \bv)\in \cH_2^1(X)$ and $u(\bx, \bv)\varphi(\bx, \bv)\in \cH_2^1(X)$, by the averaging lemma~\cite{devore2001averaging,diperna1991lp}, we have
	\begin{equation}
	\int_{\bbS^{d-1}} w(\bx, \bv) \mathfrak{W}(\bx, \bv) d\bv \in W^{2,1/2}(\Omega),\quad\int_{\bbS^{d-1}}u(\bx, \bv)\varphi(\bx, \bv) d\bv \in W^{2,1/2}(\Omega).
	\end{equation}
	Similarly, it is easy to verify that $\mathcal{K}w(\bx, \bv) \mathfrak{W}(\bx, \bv)\in \cH_2^1(X)$ and $\mathcal{K}u(\bx, \bv) \psi(\bx, \bv)\in \cH_2^1(X)$ as well, hence 
	\begin{equation}
	\int_{\bbS^{d-1}} \cK w(\bx, \bv) \mathfrak{W}(\bx, \bv) d\bv \in W^{2,1/2}(\Omega),\quad\int_{\bbS^{d-1}}\cK u(\bx, \bv)\varphi(\bx, \bv) d\bv \in W^{2,1/2}(\Omega).
	\end{equation}
	By the compactness of the embedding from $W^{2,1/2}(\Omega)$ into $L^2(\Omega)$, we obtain that ${\cS}: L^2(\Omega)\rightarrow L^2(\Omega)$ is Fredholm.
\end{proof}

\section{Reconstruction of $\sigma_{x,f}$}\label{SEC:SIGMA}
In this section, we consider the reconstruction of the coefficient $\sigma_{x,f}$ from the internal data $H$ in two scenarios.
We first show that 
%with only one illumination, 
the linearized inverse problem permits a unique reconstruction when the medium is optically thin and the scattering is weak. Then propose a proximal reconstruction for the nonlinear problem, which allows an arbitrary accuracy when $\sigma_{x,f}$ is H\"{o}lder continuous and is known near the boundary. Note that without these further assumptions on the medium parameters, this nonlinear inverse medium problem may not have a unique reconstruction, that is, two different $\sigma_{x,f}$'s may give the same $H$ (see Section~\ref{SEC:NUM}). 
%By Fredholm alternative, we derive the following theorem for the local uniqueness.

\subsection{Uniqueness and stability for linearized problem}\label{LOCAL}
We have the following theorem for the linearized problem. 
\begin{theorem}\label{THM:LUNIQ}
	Let $g(\bx)\in L^{\infty}(\partial\Omega)$, and suppose that the assumptions~($\mathfrak{A}$1-$\mathfrak{A}$4) hold. Let the following conditions be satisfied:
	\begin{enumerate}
		\item The medium is optically thin, i.e., there exists a small constant $1 > \gamma > 0$ such that 
		\begin{equation}\label{EQ:COND1}
		\exp(\ell_{\Omega}\sup_{\bx\in\Omega}\sigma_{x,tf}(\bx) ) < 1+\gamma \;\text{ with }\; \ell_{\Omega} = \text{diam}(\Omega)
		\end{equation}
		\item The scattering is weak,  i.e., there exists a small constant $1 > \delta > 0$ such that
		\begin{equation}\label{EQ:COND2}
		\sup_{\bx\in\Omega}\frac{\sigma_{x,s}}{\sigma_{x,tf}} < \delta
		\end{equation}
		\item The constants $\gamma$ and $\delta$ satisfy
		\begin{equation}\label{EQ:COND3}
		(1+\delta)(1 + 2\mu^2(1+\gamma)^2) < \frac{1+2\gamma}{\gamma}\;\text{ with }\;\mu = \sup_{\bx\in\partial\Omega} g(\bx) / \inf_{\bx\in\partial\Omega} g(\bx).
		\end{equation}
	\end{enumerate}
	Then the linear equation 
	${\opH}'[\sigma_{x,f}] \delta\sigma_{x,f} = 0$
	only permits the zero solution.
\end{theorem}

\begin{proof}
	It follows from~\eqref{EQ:V} that
	\begin{equation}
	-\delta\sigma_{x,f} u = \bv \cdot \nabla v(\bx, \bv) + \sigma_{x,tf} v(\bx, \bv) - \sigma_{x,s} \cK v(\bx, \bv).
	\end{equation}
	Substituting the above equation into~\eqref{EQ:HP}, we obtain that when 
	${\opH}'[\sigma_{x,f}](\delta\sigma_{x,f}) = 0$, $v(\bx, \bv)$ satisfies
	\begin{equation}\label{EQ:PRTE}
	\begin{aligned}
	\bv \cdot \nabla v(\bx, \bv) + \sigma_{x,tf} v(\bx, \bv)&=  \sigma_{x,s} \cK v(\bx, \bv) + 2\sigma_{x,tf}{\cK_1}v - 2\sigma_{x,s}{\cK_1}\cK v &\text{ in }&X,\\
	v(\bx, \bv) &= 0\quad&\text{ on }&\Gamma_{-},
	\end{aligned}
	\end{equation}
	where the map ${\cK_1}$ is defined as 
	\begin{equation}
	\begin{aligned}
	{\cK_1} f(\bx, \bv) = \frac{1}{\psi} \int_{\bbS^{d-1}} u(\bx,\bv) u(\bx, -\bv') f(\bx, \bv') d\bv' \;\text{ with }\;
	\psi(\bx) = \int_{\bbS^{d-1}} u(\bx, \bv) u(\bx, -\bv) d\bv.
	\end{aligned}
	\end{equation}
	Define the operators $L$ and ${T}$ by
	\begin{equation*}
	L:= \sigma_{x,tf}^{-1}\left(\bv \cdot \nabla + \sigma_{x,tf}\right),\;\text{ and }\; {T}:=  \sigma_{x,tf}^{-1}\left(\sigma_{x,s} \cK  + 2\sigma_{x,tf}{\cK_1} - 2\sigma_{x,s}{\cK_1}\cK \right).
	\end{equation*}
	Let the space $L_{\sigma}^p(X)$ be the space of functions with the norm
	\begin{equation*}
	\|f(\bx, \bv)\|_{\sigma} := \int_{X} \sigma_{x,tf}(\bx) |f(\bx, \bv)|^p d\bx d\bv.
	\end{equation*}
	By Lemma 4.1 in~\cite{vladimirov1963mathematical}, $\|L^{-1}\|_{L^p_{\sigma}(X)}\le (1 - \exp(-\ell_{\Omega}\sup_{\bx\in\Omega}\sigma_{x,tf}(\bx) ))$. We also have the following estimate for $u(\bx, \bv)$ from the maximum principle and semigroup theory, 
	\begin{equation*}
	\exp\left( -\ell_{\Omega}\sup_{\bx\in\Omega}\sigma_{x,tf}(\bx)\right)\inf_{\bx\in\partial\Omega} g(\bx)\le u(\bx, \bv) \le \sup_{\bx\in\partial\Omega} g(\bx).
	\end{equation*}
	Thus $\|{T}\|_{L^p_{\sigma}(X)}$ is bounded by
	\begin{equation*}
	\|{T}\|_{L^p_{\sigma}(X)} \le \sup_{\bx\in\Omega} \frac{\sigma_{x,s}}{\sigma_{x,tf}} + 2\mu^2 \exp\left( 2\ell_{\Omega} \sup_{\bx\in\Omega}\sigma_{x,tf}(\bx)\right) + 2\mu^2 \exp\left( 2\ell_{\Omega} \sup_{\bx\in\Omega}\sigma_{x,tf}(\bx)\right)\cdot \sup_{\bx\in\Omega} \frac{\sigma_{x,s}}{\sigma_{x,tf}}.
	\end{equation*}
	From the given conditions~\eqref{EQ:COND1},~\eqref{EQ:COND2} and~\eqref{EQ:COND3}, we deduce that $\|L^{-1}\|_{L^p_{\sigma}(X)}\le \frac{\gamma}{1+\gamma}$ and $\|{T}\|_{L^p_{\sigma}} \le \delta + 2\mu^2 (1+\gamma)^2(1+\delta)$. Thus
	\begin{equation}\label{EQ:LS}
	\begin{aligned}
	\|L^{-1} {T}\|_{L^p_{\sigma}(X)} &\le \|L^{-1}\|_{L^p_{\sigma}(X)} \|{T}\|_{L^p_{\sigma}(X)}\\
	&\le \frac{\gamma}{1+\gamma} \left( \delta + 2\mu^2 (1+\gamma)^2(1+\delta)\right)\\
	&= \frac{\gamma}{1+\gamma} \left( (1+ 2\mu^2 (1+\gamma)^2)(1+\delta) - 1\right) < 1.
	\end{aligned}
	\end{equation}
	Therefore $L v = {T}v$ only permits $v = 0$ in $L^p_{\sigma}(X)$, and the proof is completed by noticing $L^p_{\sigma}(X)$ is the same set as $L^p(X)$.
\end{proof}  
The above local uniqueness result could be interpreted by considering the limiting case. When the scattering coefficient $\sigma_{x,s}\to 0$, the internal data $H\to -\sigma_{x,tf}\int_{\bbS^{d-1}} u(\bx, \bv) u(\bx, -\bv) d\bv$. If the medium is optically thin or $\text{diam}(\Omega)\ll 1$, then the solution $u(\bx, \bv)$ could be well approximated by ignoring the coefficient $\sigma_{x,f}$ in~\eqref{EQ:SRTE}, thus $\sigma_{x,f}$ and $\sigma_{x,tf}$ are decoupled and can be recovered directly. 

The following stability estimate follows immediately from the classical stability theory of Fredholm operators~\cite{kato2013perturbation}. 
\begin{theorem}\label{THM:LOCAL STAB}
	Let $\mathfrak{H}$ and $\tilde{\mathfrak{H}}$ be two perturbed internal data defined in~\eqref{EQ:HP}, and $\delta \sigma_{x,f}$ and $\delta \tilde{\sigma}_{x,f}$ be the solutions to the linearized equations
	\begin{equation}
	{\opH}'[\sigma_{x,f}]\delta\sigma_{x,f} = \mathfrak{H}\;\text{ and }\; 	{\opH}'[\sigma_{x,f}]\delta\tilde{\sigma}_{x,f} = \tilde{\mathfrak{H}},
	\end{equation}
	where $\sigma_{x,f}$ is the background coefficient. Then under the same condition as Theorem~\ref{THM:LUNIQ}, there exists a constant $\tilde{\sfc} = \tilde{\sfc}(\gamma, \delta) > 0$ such that
	\begin{equation}
	\frac{1}{\tilde{\sfc}} \|\mathfrak{H} - \tilde{\mathfrak{H}}\|_{L^2(\Omega)} \le  \|\delta\sigma_{x,f} -  \delta\tilde{\sigma}_{x,f}\|_{L^{2}(\Omega)} \le \tilde{\sfc} \|\mathfrak{H} - \tilde{\mathfrak{H}}\|_{L^2(\Omega)}.
	\end{equation} 
\end{theorem}

\begin{remark}
{We point out that most biological tissues are typically strongly scattering, for which $\sup_{\bx\in\Omega}\frac{\sigma_{x,s}}{\sigma_{x,tf}}$ is close to $1$, thus the second assumption in Theorem 4.1 does not hold. For the conclusion in Theorem 4.1 to hold in this case, it requires the domain size $\ell_{\Omega}$ to be small enough. 
Alternatively, we introduce a proximal reconstruction method in Section 4.2, which requires neither $\sup_{\bx\in\Omega}\frac{\sigma_{x,s}}{\sigma_{x,tf}}\ll 1$ nor $\ell_{\Omega}$ is small.}
    %{\color{blue} Let us point out most biological tissues are typically strongly scattering, the second assumption in Theorem \ref{THM:LUNIQ} has $\sup_{\bx\in\Omega}\frac{\sigma_{x,s}}{\sigma_{x,tf}}$  close to $1$, then the result will require the domain size $\ell_{\Omega}$ to be small enough. However, we introduce another proximal reconstruction method in Section~\ref{SEC:GLOBAL} which does not require such condition.}
\end{remark}
\subsection{Proximal uniqueness and stability of nonlinear problem}\label{SEC:GLOBAL}
We now show an approach to 
 approximating the coefficient $\sigma_{x,f}$ with arbitrary accuracy when $\sigma_{x,f}$ is H\"{o}lder continuous and is known near the boundary. For preparation, we introduce the following definitions and lemma. 
\begin{definition}[$\delta$-covering]
	Let $(M, \dist)$ be a metric space. The set $V$ is a $\delta$-covering of $M$ if for every $\bx \in M$, there exists $\by\in V$ such that $\dist(\by, \bx) \le \delta$.
\end{definition}
\begin{definition}[$\delta$-packing]
	Let $(M, \dist)$ be a metric space. The set $V$ is a $\delta$-packing of $M$ if for every $\bx_1\neq\bx_2 \in V$, $\dist(\bx_1, \bx_2) > \delta$.
\end{definition}
\begin{definition}\label{DEF:SKELETON}
	Suppose $\Omega\subset \bbR^d$ is a convex domain and $\dist$ is a metric defined on $\bbR^d$. Let $V = \{{\by}_i\}_{i=1}^n$ be a vertex set with $\by_i\in\partial\Omega$, and $G = G(V)$ be the complete geometric graph formed from the vertices $V$. Denote the edge set of $G$ by $E$. For every $e\in E$, we define the $\theta$-tube of $e$ by
	\begin{equation}
	T_{\theta}(e) := \{\by \in  \bbR^{d}: \dist(\by, e) <\theta \}.
	\end{equation}
	We then define the $\theta$-skeleton of the graph $G$ by
	\begin{equation}
	G_{\theta}(V) := G(V) \setminus \bigcup_{e_1\neq e_2\in E} \left( T_{\theta}(e_1)\cap e_2\right).  
	\end{equation}
	See Fig~\ref{FIG:SKELETON} for an illustration of the formation of $G_{\theta}(V)$.
\end{definition}
\begin{figure}[!htb]
	\centering
	\includegraphics[scale=0.13]{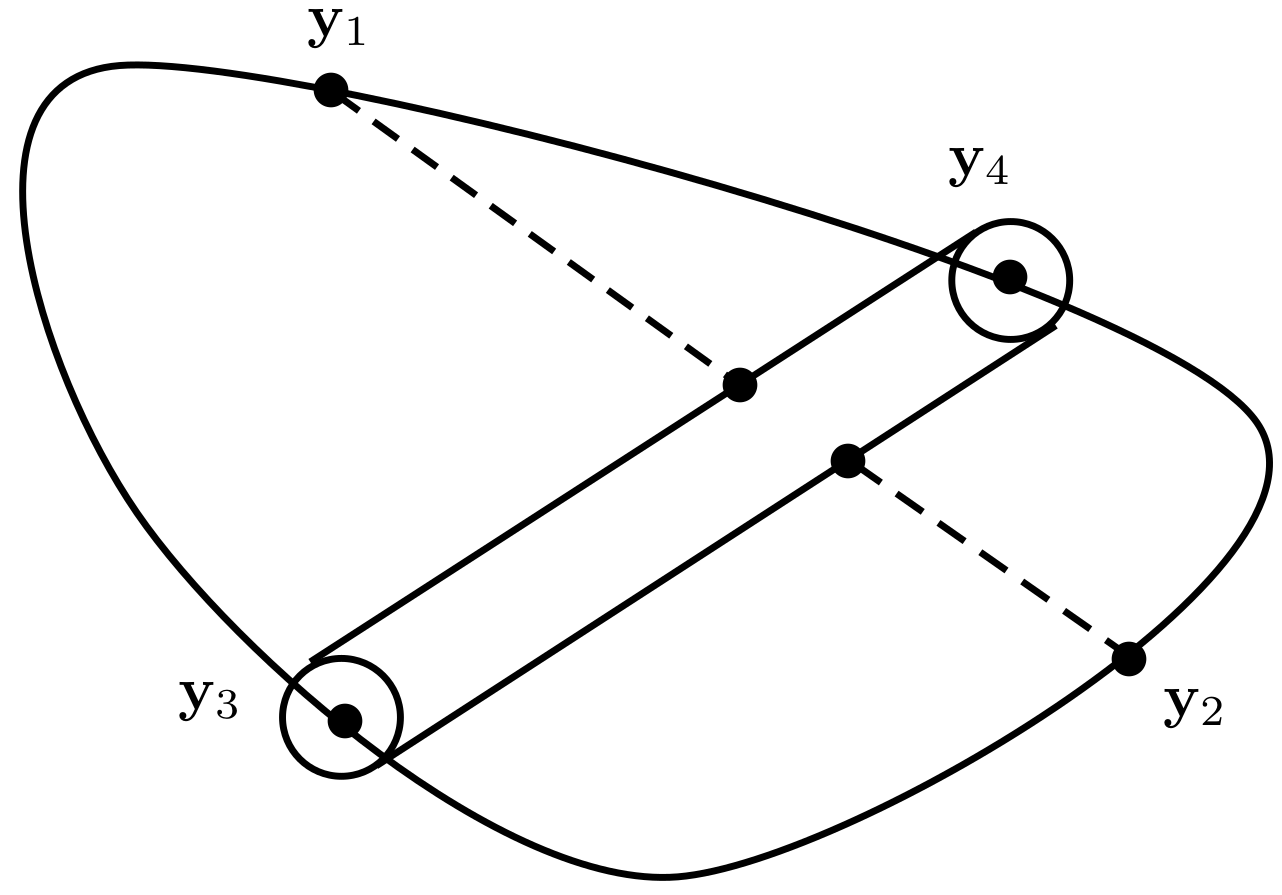}
	\caption{Illustration of the $\theta$-skeleton of a graph $G$.  The vertex set is $V = \{\by_1, \by_2, \by_3, \by_4\}$.  The intersection of the segment $\by_1\by_2$ with the $\theta$-tube of the edge $\by_3\by_4$ is removed from $G(V)$. }\label{FIG:SKELETON}
\end{figure}

\begin{lemma}\label{LEM:NET}
	Suppose $\Omega\subset \bbR^{d}$ is a unit ball and $\dist$ is the Euclidean metric defined on $\bbR^d$. Then for sufficiently small $\delta >0$, there exists a vertex set $V=\{ \by_i\}_{i=1}^n\subset \partial\Omega$ with $n = \cO(\delta^{1-d})$ such that the $\theta$-skeleton $G_{\theta}(V)$ generated by $V$ is a $2\delta$-covering of $\Omega$ for sufficiently small $\theta \le \cO(\delta^2/n^2)$, i.e., for any point $\bx \in \Omega$, there exists a point $\by \in G_{\theta}(V)$ such that $\dist(\bx, \by)\le 2\delta$.
\end{lemma}

\begin{proof}
	Given any $\delta > 0$, we choose a $\delta$-packing $V$ of $\partial\Omega$ with maximal cardinality. It follows that $V$ is also a $\delta$-covering of $\partial\Omega$ and $\text{card}(V) =\cO(\delta^{1-d})$. We claim that  $G(V)$ forms a $\delta$-covering of $\Omega$. For any $\bx \in\Omega$, we pick an arbitrary point $\by_i\in V$, and denote the other intersection of $\partial\Omega$ and the line through $\by_i$ and $\bx$ by $\bx'$. Since $V$ is a $\delta$-covering of $\partial\Omega$, there exists a point $\by_j\in V$ such that $\dist(\bx', \by_j) \le \delta$. When $\by_i \neq \by_j$, we have that $\dist(\bx, e_{ij}) \le \dist(\bx', \by_j)\le \delta$, where $e_{ij}$ is the edge connecting the vertices $\by_i$ and $\by_j$ (see Fig~\ref{FIG:LEM SKELETON}). The claim is obviously true when $\by_i \neq \by_j$.
	\begin{figure}[!htb]
		\centering
		\includegraphics[scale=0.18]{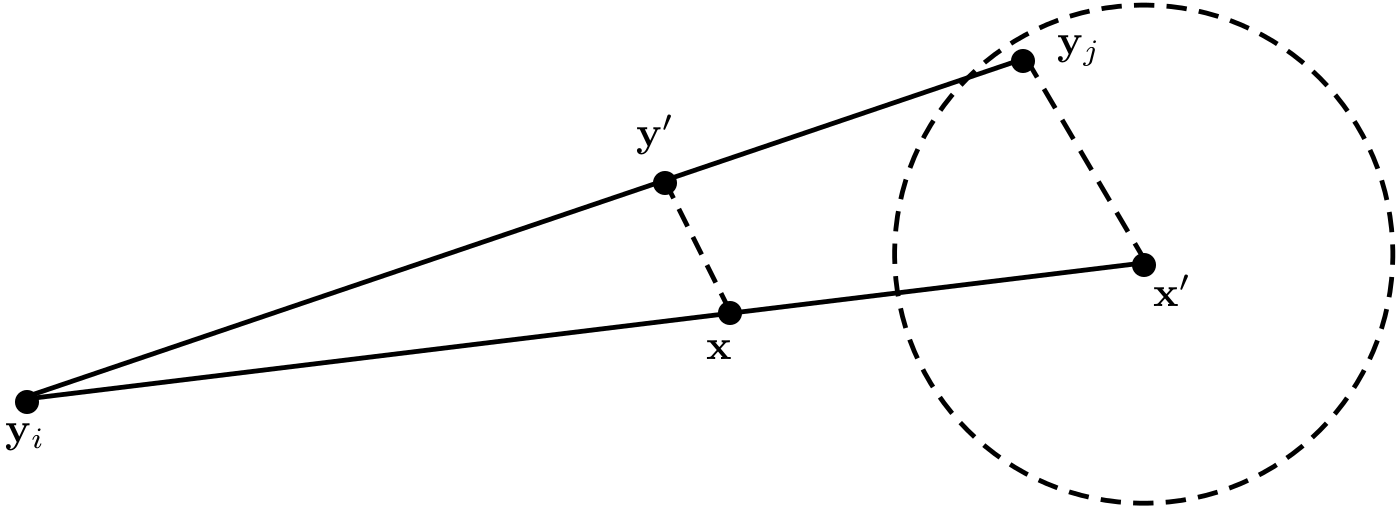}
		\caption{Illustration of the claim that $G(V)$ forms a $\delta$-covering of $\Omega$. The dashed circle is centered at $\bx'$ and has radius $\delta$. The point $\by'$ is on the line $\by_i\by_j$, and $\bx\by'$ is parallel to $\bx'\by_j$. We have $\by'\in\Omega$ by the convexity of $\Omega$, and $\dist(\bx, \by') \le \dist(\by_j, \bx')\le \delta$ by the similarity between the triangles $\triangle(\by_i\by_j\bx')$ and $\triangle(\by_i\by'\bx)$.}\label{FIG:LEM SKELETON}
	\end{figure}
	
	Let $\theta \le \frac{\delta^2}{4 n^2}$. Consider an edge $e_1\in E$. For any $e_2\in E$ and $e_2\neq e_1$, the length of $T_{\theta}(e_1)\cap e_2$ is at most $2\theta/\sin\alpha$, where $\alpha\in (0, \frac{\pi}{2})$ is the angle between $e_1, e_2$. On the other hand, since $\Omega$ is the unit ball and $|e_1|> \delta$ from the fact that $V$ is a $\delta$-packing of $\partial\Omega$, we must have $\sin\alpha \ge \frac{\delta}{2}$, therefore $|T_{\theta}(e_1)\cap e_2|\le \frac{4\theta}{\delta}$. Because $\text{card}(E) = \binom{n}{2} \le n^2$, the total length removed from $e_1$ is at most $n^2 \frac{4\theta}{\delta} \le \delta$. 
	
	We now prove that $G_{\theta}(V)$ is a $2\delta$-covering of $\Omega$. For any $\bx \in \Omega$, since $G(V)$ is a $\delta$-covering of $\Omega$, we can find $\by \in G(V)$ such that $\dist(\bx, \by) \le \delta$ and an edge $e \in E$ such that $\by \in e$. Because the total length removed from $e$ is at most $\delta$, we can find a point $\bt \in e\cap G_{\theta}(V)$ such that $\dist(\bt, \by) \le \delta$ and $\dist(\bx, \bt) \le \dist(\bx, \by) + \dist(\by, \bt) \le 2\delta$.
\end{proof} 
We remark that $\Omega$ is taken to be a unit ball in the above lemma only for simplicity. The proof can be easily adapted to the case when the principal curvatures of $\partial\Omega$ are bounded away from zero. In the following, we assume that the domain $\Omega$ is the unit ball and prove the global uniqueness by a constructive method. The idea is to use the fact that the quadratic term $\int_{\bbS^{d-1}} u(\bx, \bv) u(\bx, -\bv) d\bv$ contains certain ``singularities" when $g(\bx)$ is concentrated at a few points on the surface $\partial\Omega$.
\begin{theorem}\label{THM:UNIQ NONLINEAR}
	Let $\Omega$ be the unit ball in $\bbR^d$ and $\dist$ the Euclidean metric defined on $\bbR^d$. Suppose the assumptions ($\mathfrak{A}$1-$\mathfrak{A}$4) hold.  Let the coefficient $\sigma_{x,f}$ satisfy the following conditions:
	\begin{enumerate}
		\item $\sigma_{x,f}$ is $\alpha$-H\"{o}lder continuous, that is, there exists a constant $\kappa > 0$,  such that $\forall \bx,\by\in\Omega$, 
		\begin{equation}
		|\sigma_{x,f}(\bx) - \sigma_{x,f}(\by)| \le \kappa \dist(\bx, \by)^{\alpha}.
		\end{equation} 
		\item We can decompose $\sigma_{x,f} = \sigma_{x,f}^0 + \delta \sigma_{x,f}$, where $\sigma_{x,f}^0$ is the known background coefficient and the unknown $\delta\sigma_{x,f}$ has compact support in the interior subdomain $\Omega_r$ for some $r \in (0, 1)$. Here
		\begin{equation}
		\Omega_r := \{ \bx : \bx \in \Omega\;\text{and}\;\dist(\bx, \partial\Omega)\ge r\}.
		\end{equation}
	\end{enumerate}
	Then for any sufficiently small $\delta > 0$, we can choose an illumination source $g(\bx)\in L^{\infty}(\partial\Omega)$ such that the internal data $H$ permits a reconstruction $\widetilde{\sigma}_{x,f}$ such that $\forall \bx\in\Omega$,
	\begin{equation}
	|\tilde{\sigma}_{x,f}(\bx) - \sigma_{x,f}(\bx)| \le \cO(\delta^{\alpha}).
	\end{equation}
\end{theorem}

\begin{proof}
	Let $\delta\ll r/5$, $n = \cO(\delta^{1-d})$ and $\theta = \cO(\delta^{2d})$. We construct a vertex set $V= \{\by_j\}_{j=1}^n\subset \partial\Omega$ whose $\theta$-skeleton $G_{\theta}(V)$ forms a $2\delta$-covering of $\Omega$ as in Lemma~\ref{LEM:NET}.
	Let $B(\bx, s)$ denote the ball centered at $\bx$ with radius $s$. We consider an illumination source function $g_{h}(\bx)$ of the form
	\begin{equation}\label{EQ:Gh}
	g_{h}(\bx) = \sum_{j=1}^n \frac{1}{h^{l}}\chi_{D_j}(\bx),\quad  \bx\in\partial\Omega,
	\end{equation}
	where $D_j = B(\by_j, h)\cap \partial\Omega$, $\chi_{D_j}$ is the characteristic function of $D_j$, the exponent $l = (d-1)/2$, and the parameter $h \le \theta$ is sufficiently small such that $\{D_j\}_{j=1}^n$ are disjoint from each other. Then $G_h\supset G_{\theta}$ is also a $2\delta$-covering of $\Omega$ and $g_h\in L^{\infty}(\partial\Omega)$ for any $h > 0$.
	
	Define the operators $\cB:L^{\infty}(\partial\Omega)\rightarrow L^{\infty}(X)$ and $\cT:L^{\infty}(X)\rightarrow L^{\infty}(X)$  as
	\begin{equation}
	\begin{aligned}
	\cB f(\bx, \bv) &= f(\bx - \tau_{-}(\bx, \bv) \bv) \exp\left( -\int_{0}^{\tau_{-}(\bx, \bv)} \sigma_{x,tf}(\bx - s\bv ) ds\right),\\
	\cT f(\bx, \bv) &= \int_0^{\tau_{-}(\bx, \bv)} \exp\left( -\int_0^l \sigma_{x,tf} (\bx - s\bv) ds\right) \sigma_{x,s}(\bx - l\bv)\cK f(\bx-l\bv, \bv) dl.
	\end{aligned}
	\end{equation}
	The solution to the RTE~\eqref{EQ:SRTE} with boundary illumination source $g_h$ is
	\begin{equation}
	u_h(\bx, \bv) = \cB g_h + \cT \cB g_h +(I - \cT)^{-1}  \cT^2\cB g_h.
	\end{equation}
	Here $\cB g_h$ is the ballistic part of the solution, $\cT \cB g_h$ is the single scattering part, and $(I - \cT)^{-1}\cT^2\cB g_h$ is the multiple scattering part. For each point $\bx\in\Omega_{r-4\delta}\subset \Omega_{r/5}$, we have that
	\begin{equation}
	\begin{aligned}
	\cK \cB g_h(\bx, \bv) &\le \int_{\bbS^{d-1}} p(\bv\cdot \bv') g_h(\bx - \tau_{-}(\bx, \bv') \bv') d\bv' \le \sfc_9 \int_{\bbS^{d-1}} g_h(\bx - \tau_{-}(\bx, \bv') \bv') d\bv'\\
	&= \sfc_9 \sum_{j=1}^n \int_{\bbS^{d-1}}\frac{1}{h^{l}} \chi_{D_j}(\bx - \tau_{-}(\bx, \bv)\bv) d\bv  \\&\le \sfc_9 \sum_{j=1}^n \frac{1}{h^{l}} \cO\left(\frac{h}{r-4\delta}\right)^{d-1}  = \cO\left( n h^{(d-1)/2}  \right).
	\end{aligned}
	\end{equation}
	It follows that $\cT \cB g_h \le \cO(n h^{(d-1)/2})$ and $(I -\cT)^{-1} \cT^2 \cB g_h \le \cO(n h^{(d-1)/2})$. Hence for all $\bx \in \Omega_r$, we can write 
	\begin{equation}
	u_h(\bx, \bv) = \cB g_h(\bx, \bv) + \cO(n h^{(d-1)/2}).
	\end{equation}
	The internal data $H_h$ is
	\begin{equation}\label{EQ:HH}
	\begin{aligned}
	H_h(\bx) &= -\sigma_{x,tf} \int_{\bbS^{d-1}} u_h(\bx, \bv) u_h(\bx, -\bv) d\bv  + \sigma_{x,s} \int_{\bbS^{d-1}} \cK u_h(\bx, \bv) u_h(\bx, -\bv) d\bv \\
	&= -\sigma_{x,tf} \int_{\bbS^{d-1}} \cB g_h(\bx, \bv)\cB g_h(\bx, -\bv) d\bv + \cO(n^2 h^{d-1}) \\
	&= -\sigma_{x,tf}\int_{\bbS^{d-1}}\mathfrak{g}_h(\bx, \bv) \exp\left({-\int_{-\tau_{+}(\bx, \bv)}^{\tau_{-}(\bx,\bv)} \sigma_{x,tf}(\bx - s\bv)ds}\right)d\bv + \cO(n^2 h^{d-1}), 
	\end{aligned}
	\end{equation}
	where $\mathfrak{g}_h(\bx, \bv) = g_h(\bx - \tau_{-}(\bx, \bv) \bv) g_h(\bx +\tau_{+}(\bx, \bv)\bv)$. Let $\by = \bx + \tau_{+}(\bx, \bv)\bv$ and $\by' = \by -\tau_{-}(\by, \bv) \bv$ with $\bv = \frac{\by - \bx}{|\by - \bx|}$. Utilizing the transformation 
	$$d\bv = \frac{1}{\nu_{d-1}}\frac{|\bn(\by)\cdot \bv|}{|\bx - \by|^{d-1}}dS_{\by}\quad $$ 
	with $\bn(\by)$ being the unit normal vector at $\by$, and noticing
	\begin{equation}
	\int_{\Omega} \chi_{D_i}(\bx \pm \tau_{\pm}(\bx, \bv) \bv) d\bv = \cO(h^{d-1}),\quad\forall\bx\in\Omega_r,
	\end{equation}
	we can rewrite the formulation~\eqref{EQ:HH} as
	\begin{equation*}
	\begin{aligned}
	H_h(\bx) &= -\sigma_{x,tf} \frac{1}{\nu_{d-1}}\int_{\partial\Omega} g_h(\by) g_h(\by') \frac{E\left(\by, \by' \right)}{|\bx - \by|^{d-1}} \left|\bn(\by)\cdot \frac{\bx - \by}{|\bx  - \by|}\right| dS_{\by} + \cO(n^2 h ^{d-1}) \\
	&= -\sigma_{x,tf}\sum_{i,j=1}^n \frac{1}{h^{d-1}} \frac{1}{\nu_{d-1}}\int_{D_{ij}(\bx)}  \frac{E\left(\by, \by' \right)}{|\bx - \by|^{d-1}} \left|\bn(\by)\cdot \frac{\bx - \by}{|\bx  - \by|}\right| dS_{\by}  + \cO(n^2 h^{d-1}),
	\end{aligned}
	\end{equation*}
	where $D_{ij}(\bx) = \{\by:\by\in D_i, \by - \tau_{-}(\by, \bv) \bv\in D_j\;\text{with}\; \bv = \frac{\by - \bx}{|\by - \bx|}  \}$ and $E(\by, \by')$ is
	\begin{equation}
	E(\by, \by') = \exp\left(-|\by - \by'| \int_0^1 \sigma_{x,tf}(\by + s(\by' -\by)) ds \right).
	\end{equation}
	Next, since $|\by - \by_i|\le h$ and $|\by - \tau_{-}(\by, \bv)\bv - \by_j|\le h$, we take the Taylor expansion at $\by= \by_i$ and $\by'=\by_j$ for each integral over $D_{ij}$ and obtain
	\begin{equation}
	\begin{aligned}
	H_h(\bx) &=  -\sigma_{x,tf}\sum_{i\neq j}^n \frac{1}{h^{d-1}} \frac{1}{\nu_{d-1}}\int_{D_{ij}(\bx)}  \left(\frac{E\left(\by_i, \by_j \right)}{|\bx - \by_i|^{d-1}} \left|\bn(\by_i)\cdot \frac{\by_i - \by_j}{|\by_i  - \by_j|}\right| +\cO(h)\right) dS_{\by}  + \cO(n^2 h^{d-1}) \\
	&=  -\sigma_{x,tf}\sum_{i\neq j}^n \frac{1}{h^{d-1}} \frac{1}{\nu_{d-1}}\int_{D_{ij}(\bx)} \frac{E\left(\by_i, \by_j \right)}{|\bx - \by_i|^{d-1}} \left|\bn(\by_i)\cdot \frac{\by_i - \by_j}{|\by_i  - \by_j|}\right| dS_{\by}  + \cO(n^2 h).
	\end{aligned}
	\end{equation}
	
	\begin{figure}[!htb]
		\centering
		\includegraphics[scale=0.18]{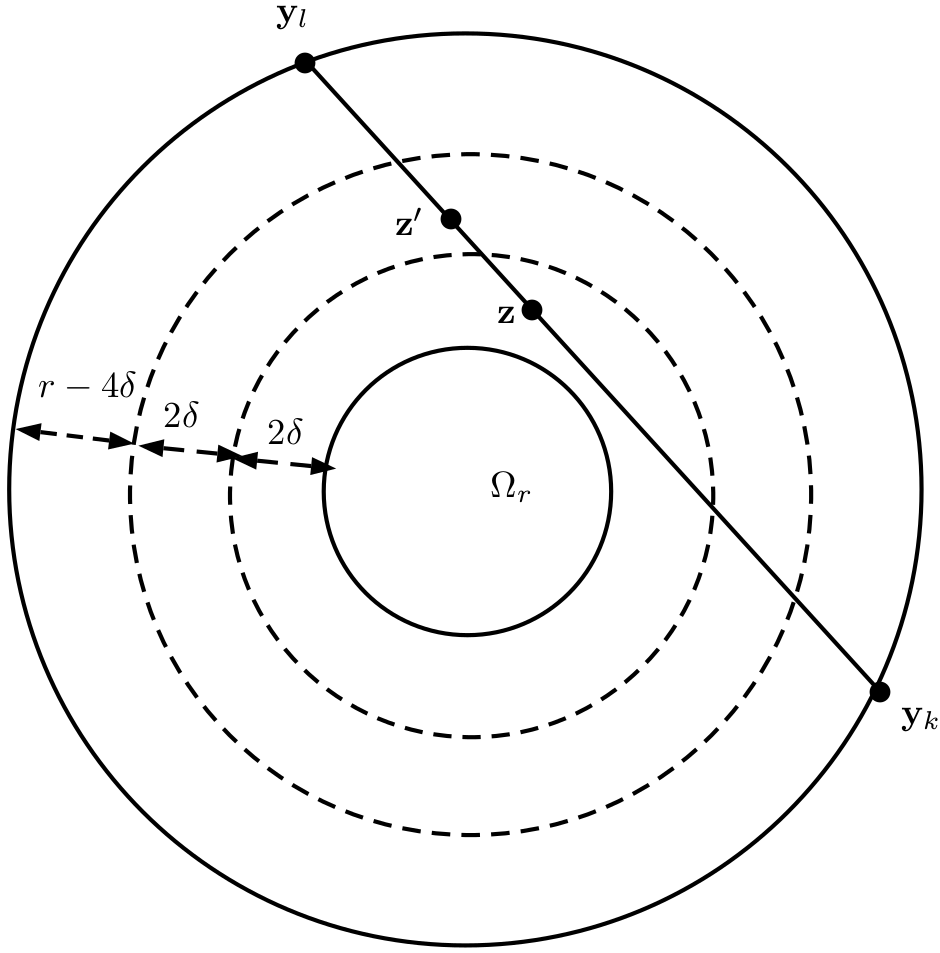}
		\caption{Illustration of \eqref{EQ:Z2}. $\Omega_r$ is the inner most ball. Let $\bz$ be a point on $ e_{lk}$ in $\Omega_{r-2\delta}\cap G_h(V)$. Since there is at most a total length of $\delta$ removed from $e_{lk}$ and the middle ring has a width of $2\delta$, we can find  $\bz'\in e_{lk}\cap \Omega_{r-4\delta}\cap \Omega_{r-2\delta}^{\complement}\cap G_h(V)$ such that \eqref{EQ:Z2} holds.  }\label{FIG:THM UNIQUE}
	\end{figure}
	For an arbitrary $\bz \in \Omega_{r-2\delta}\cap G_{h}(V)$, there exists a unique edge $e_{lk} \in E$ connecting $\by_l$ and $\by_k$ such that $\bz = t\by_k + (1-t)\by_l$ for some $t\in (0, 1)$. This means that $D_{ij}(\bz)\neq \emptyset$ if and only if $i=k, j=l$ or $i=l, j=k$. Therefore we have
	\begin{equation}
	H_h(\bz) = -\sigma_{x,tf}(\bz) E(\by_k, \by_l) \mathfrak{B}_h(\bz, \by_k, \by_l) + \cO(n^2 h),
	\end{equation}
	where $\mathfrak{B}_h(\bz, \by_k,\by_l)$ satisfies that for some constant $\sfc_{10}>0$,
	\begin{equation}\label{EQ:LOW}
	\begin{aligned}
	\mathfrak{B}_h(\bz, \by_k, \by_l) &= \frac{1}{h^{d-1}\nu_{d-1}}\left(\int_{D_{kl}(\bz)} \frac{\left|\bn(\by_k)\cdot \frac{\by_k - \by_l}{|\by_k -\by_l|} \right|}{|\bz - \by_k|^{d-1}}  dS_{\by} + \int_{D_{lk}(\bz)} \frac{\left|\bn(\by_l)\cdot \frac{\by_k - \by_l}{|\by_k -\by_l|} \right|}{|\bz - \by_l|^{d-1}}  dS_{\by}\right) \\
	&= \frac{|\by_k - \by_l|}{2\nu_{d-1}}\frac{1}{h^{d-1}}\left(\int_{D_{kl}(\bz)} \frac{1}{|\bz - \by_k|^{d-1}}  dS_{\by} + \int_{D_{lk}(\bz)} \frac{1}{|\bz - \by_l|^{d-1}}  dS_{\by}\right)\\
	&\ge \sfc_{10}|\by_k - \by_l|.
	\end{aligned}
	\end{equation}
	Here we have used $\bn(\by_k) = \by_k$ in the second equality. In general, if the principal curvatures of $\partial\Omega$ are bounded away from zero, the same lower bound in~\eqref{EQ:LOW} still holds. Since $\cB_h(\bz, \by_k, \by_l)\geq \sfc_{10} \delta$ for some $\sfc_{10} >0$ and $\cB_h(\bz, \by_k, \by_l)$ is independent of $\sigma_{x,f}$, we have
	\begin{equation}\label{EQ:Z1}
	\sigma_{x,tf}(\bz) = -\frac{H_h(\bz)}{E(\by_k, \by_l)\cB_h(\bz, \by_k, \by_l)} + \cO\left(\frac{n^2 h}{\delta}\right).
	\end{equation}
	On the other hand, by the argument in the proof of Lemma~\ref{LEM:NET}, there exists $\bz'\in \Omega_{r-2\delta}^{\complement}\cap \Omega_{r-4\delta} \cap e_{lk}\cap G_{h}(V)$ (see Fig~\ref{FIG:THM UNIQUE}) such that
	\begin{equation}\label{EQ:Z2}
	\sigma_{x,tf}(\bz') = -\frac{H_h(\bz')}{E(\by_k, \by_l)\cB_h(\bz', \by_k, \by_l)} + \cO\left(\frac{n^2 h}{\delta}\right),
	\end{equation}
	which is known from the background coefficient $\sigma_{x,f}^0$. Taking the ratio between~\eqref{EQ:Z1} and~\eqref{EQ:Z2}, we obtain
	\begin{equation}\label{EQ:sxf}
	\sigma_{x,tf}(\bz)=\sigma_{x,tf}(\bz') \frac{H_h(\bz)}{H_h(\bz')} \frac{\cB_h(\bz', \by_k, \by_l)}{\cB_h(\bz, \by_k, \by_l)} + \cO\left(\frac{n^2 h}{\delta}\right).
	\end{equation}
	Recalling that $n = \cO(\delta^{1-d})$ and $h \le \theta = \cO(\delta^{2d})$, we have $\cO\left(\frac{n^2 h}{\delta}\right)\le\cO(\delta)$. Therefore for each $\bz\in \Omega_{r-2\delta}\cap G_{h}(V)$, we can recover $\sigma_{x,tf}$ (hence $\sigma_{x,f}$) up to an error of $\cO(\delta)$. Since $\Omega_{r-2\delta}\cap G_{h}(V)$ is a $2\delta$-covering for $\Omega_r$, for any $\bx \in \Omega_r$, we can find a point $\bz \in \Omega_{r-2\delta}\cap G_{h}(V)$ such that $|\bz - \bx|\le 2\delta$. Using the H\"{o}lder continuity condition of $\sigma_{x,f}$, we conclude that the $L^{\infty}$ reconstruction error is bounded by $\cO(\delta) + \kappa (2\delta)^{\alpha} = \cO(\delta^{\alpha})$. 
\end{proof}
Notice that, if the conditions in Theorem~\ref{THM:UNIQ NONLINEAR} are not satisfied, then the uniqueness of the above nonlinear case might not hold under certain circumstances.
We demonstrate a numerical example which permits two distinct reconstructions for this situation in Example~\ref{EX:1}. In practice, the specific singular illumination source in~\eqref{EQ:Gh} with $h\to 0$ is not possible due to resolution limitation. However, for a moderately small~$h$, and a source $g_h$ which only concentrates at a few spots on the boundary, and when the total absorption coefficient $\sigma_{x,tf}$ is not too large, the ballistic signal still can be captured in $\int_{\bbS^{d-1}} u(\bx, \bv) u(\bx, -\bv) d\bv$ near its $h$-skeleton.  This could be used to recover the information on the $h$-skeleton approximately; see Fig~\ref{FIG:DEMO}. Although the uniqueness result of the above theorem is ``proximal'' and constructive, it does not rule out uniqueness for other types of illumination source.
\begin{figure}[!htb]
	\centering
	\includegraphics[height=0.3\textwidth]{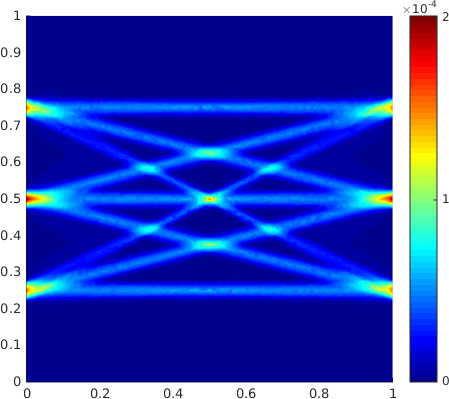}\qquad
	\includegraphics[height=0.3\textwidth]{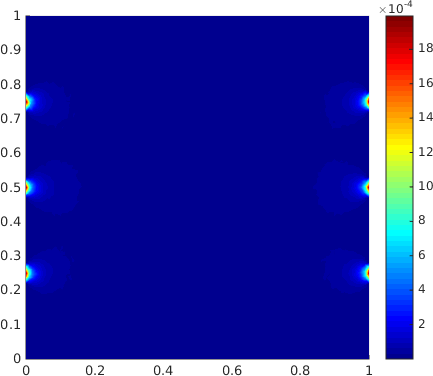}
	\caption{The left side is $\intS u(\bx, \bv)u(\bx, -\bv) d\bv$ and the right side is $\intS \cK u(\bx, \bv) u(\bx, -\bv) d\bv $. Here the solution $u(\bx, \bv)$ to~\eqref{EQ:SRTE} is solved using homogeneous coefficients $\sigma_{x,a}\equiv 0.2$, $\sigma_{x,s}\equiv 0.2$, $\sigma_{x,f}\equiv 0.5$ and isotropic scattering. The boundary illumination ~\eqref{EQ:Gh} consists of six points on the two sides with $h = \frac{1}{32}$. }\label{FIG:DEMO}
\end{figure}

\begin{theorem}\label{THM:STAB NONLINEAR}
	Under the same assumptions of Theorem~\ref{THM:UNIQ NONLINEAR}, let $ \sigma^1_{x,f}$ and $\sigma^2_{x,f}$ be two admissible absorption coefficients of the fluorophores. Choose the illumination source $g_h$ constructed in Theorem~\ref{THM:UNIQ NONLINEAR} with $\delta \ll 1$, suppose $H_1, H_2$ are the corresponding internal data associated with $ \sigma^1_{x,f}$ and $\sigma^2_{x,f}$ respectively. If $H_1$ and $H_2$ satisfy
	$
	\left\| H_1/H_2 - 1\right\|_{L^{\infty}(\Omega_{r-4\delta})} <1
	$,
	then
	\begin{equation}
	\|\sigma^1_{x,f} - \sigma^2_{x,f}\|_{L^{\infty}(\Omega)} \le \cO\left( \delta^{\alpha} + \|H_1/H_2 - 1\|_{L^{\infty}(\Omega_{r-4\delta})}\right).
	\end{equation}
\end{theorem}
\begin{proof}
	Using the same argument as in Theorem~\ref{THM:UNIQ NONLINEAR}, 
	for any $\bz \in \Omega_{r - 2\delta}\cap G_h(V)$, there is a unique edge $e_{lk}\in E$ such that $\bz \in e_{lk}$ and we can find $\bz'\in G_h(V) \cap \Omega_{r-4\delta}\cap \Omega_{r-2\delta}^{\complement}\cap e_{lk}$ such that
	\begin{equation}
	\begin{aligned}
	\sigma^1_{x,tf}(\bz)= \sigma^1_{x,tf}(\bz') \frac{H_1(\bz)}{H_1(\bz')} \frac{\cB_h(\bz', \by_k, \by_l)}{\cB_h(\bz, \by_k, \by_l)} + \cO\left(\delta\right),\\
	\sigma^2_{x,tf}(\bz)=\sigma^2_{x,tf}(\bz') \frac{H_2(\bz)}{H_2(\bz')} \frac{\cB_h(\bz', \by_k, \by_l)}{\cB_h(\bz, \by_k, \by_l)} + \cO\left(\delta\right),
	\end{aligned}
	\end{equation}
	where $\sigma^i_{x,tf} = \sigma_{x,a} +\sigma_{x,s} + \sigma^i_{x,f}$ for $i=1,2$. Taking the ratio of the above two equations and using the fact that $\sigma^i_{x,f}$ is known outside the subdomain $\Omega_r$, we obtain
	\begin{equation}
	\left|\frac{\sigma^1_{x,tf}(\bz)}{\sigma^2_{x,tf}(\bz)} - 1\right| = \left| \frac{H_1(\bz)}{H_2(\bz)} \frac{H_2(\bz')}{H_1(\bz')} - 1 +\cO(\delta)\right| \le \frac{2 \|H_1/H_2 - 1\|_{L^{\infty}(\Omega_{r-4\delta})}}{1 -  \|H_1/H_2 - 1\|_{L^{\infty}(\Omega_{r-4\delta})}} +\cO(\delta) .
	\end{equation}
	Since $\|H_1/H_2 - 1\|_{L^{\infty}(\Omega_{r-4\delta})} <1$, we obtain the error estimate for any $\bz \in \Omega_{r-2\delta}\cap G_h(V)$,
	\begin{equation}
	\left|\sigma^1_{x,f}(\bz) - \sigma^2_{x,f}(\bz)\right| \le \cO\left( \|H_1/H_2 - 1\|_{L^{\infty}(\Omega_{r-4\delta})} +\delta \right).
	\end{equation}
	Again, since $(\sigma^1_{x,f} - \sigma^2_{x,f})$ is  $\alpha$-H\"{o}lder continuous and $\Omega_{r-2\delta}\cap G_h(V)$ is a $2\delta$-covering of $\Omega_r$, we obtain
	\begin{equation}
	\|\sigma^1_{x,f} - \sigma^2_{x,f}\|_{L^{\infty}(\Omega)} \le \cO\left(  \|H_1/H_2 - 1\|_{L^{\infty}(\Omega_{r-4\delta})} + \delta^{\alpha} \right).
	\end{equation}
\end{proof}

\section{Reconstruction of $\eta$}\label{SEC:ETA}
Once $\sigma_{x,f}$ is reconstructed from the internal data $H$ at the excitation stage, we use the reconstructed coefficient $\sigma_{x,f}$ to reconstruct the quantum efficiency using the internal data $S$ at the emission stage. In practice, the reconstruction of $\sigma_{x,f}$ cannot be exact due to measurement noise. In the following theorem, we show that, as long as the error of the reconstructed $\sigma_{x,f}$ is controlled and a mild invertibility condition is satisfied, the error of the reconstructed $\eta$ is also controlled. This result can be understood by regarding the error in $\sigma_{x,f}$ as a perturbation of a compact operator, where the eigenvalues vary continuously with the perturbation~\cite{dunford1963linear}.
\begin{theorem}\label{THM:STAB ETA}
	Let $g(\bx)\in L^{\infty}(\Omega)$ and suppose that the assumptions ($\mathfrak{A}$1-$\mathfrak{A}4$) hold. Suppose $(\sigma_{x,f}, \eta), (\tilde{\sigma}_{x,f}, \tilde{\eta})\in \cA_{\sigma}\times \cA_{\eta}$ are two pairs of admissible coefficients and $\|\sigma_{x,f} - \tilde{\sigma}_{x,f}\|_{L^{\infty}(\Omega)}\le \eps'$ is sufficiently small.  Let $(u,w)$ and $(\tilde{u}, \tilde{w})$ be the solutions for the coefficients $(\sigma_{x,f},\eta)$ and $(\tilde{\sigma}_{x,f}, \tilde{\eta})$ respectively,  and let $S$ and $\tilde{S}$ be the corresponding internal data at emission stage for  $(\sigma_{x,f}, \eta)$ and $(\tilde{\sigma}_{x,f}, \tilde{\eta})$ respectively. Define the linear operators 
	\begin{equation*}
	\begin{aligned}
	\cL_x& := \bv\cdot \nabla + \sigma_{x,tf}, &\cS_x:= \sigma_{x,s}\cK, \\
	\cL_m&:= \bv \cdot \nabla +( \sigma_{m,s}+\sigma_{m,a}),\quad &\cS_m:= \sigma_{m,s}\cK.
	\end{aligned}
	\end{equation*}
	We then define the linear operators $\cA_i:L^2(\Omega)\to L^2(\Omega)$, $1\le i\le 3$ as
	\begin{equation}
	\begin{aligned}
	\cA_1 f &= \sigma_{x,f} (\cI u) (\cI \fW) f, \\
	\cA_2 f &= -(\sigma_{m,a} +\sigma_{m,s})\intS (I - \cL_m^{-1}\cS_m)^{-1} ( \cL_m^{-1} (\sigma_{x,f}(\cI u) f)) (\bx, \bv) \fW(\bx, \bv)d\bv \\&\quad+ \sigma_{m,s}\intS \cK(I - \cL_m^{-1}\cS_m)^{-1} ( \cL_m^{-1} (\sigma_{x,f}(\cI u) f))(\bx, \bv) \fW(\bx, \bv)d\bv,\\
	\cA_3 f&= -\sigma_{x,tf}\intS u(\bx, \bv)(I - \cL_x^{-1}\cS_x)^{-1} \cL_x^{-1} ( \sigma_{x,f}(\cI\fW) f ) (\bx, \bv) d\bv \\
	&\quad +\sigma_{x,s}\intS \cK u(\bx, \bv) (I - \cL_x^{-1}\cS_x)^{-1} \cL_x^{-1} (\sigma_{x,f}(\cI\fW) f)(\bx, \bv) d\bv,
	\end{aligned}
	\end{equation}
where $\fW$ is defined in equation~\eqref{EQ:W}. If zero is not an eigenvalue of $\cA_1 + \cA_2 + \cA_3$, then there exists a constant $C > 0$ such that
	\begin{equation}
	\|\eta - \tilde{\eta}\|_{L^2(\Omega)} \le C (\|S - \tilde{S}\|_{L^2(\Omega)} + \eps').
	\end{equation}
\end{theorem}
\begin{proof}
	Let $\delta u(\bx, \bv)$, $\delta{w}(\bx, \bv)$ and $\delta \varphi(\bx, \bv)$ be the solutions to the following RTEs, 
	\begin{equation}
	\begin{aligned}
	\bv\cdot \nabla \delta{u}(\bx, \bv) + \sigma_{x,tf} \delta{u}(\bx, \bv) &= \sigma_{x,s}\cK \delta{u}(\bx, \bv) - \delta{\sigma}_{x,f}\tilde{u}\qquad&\text{ in }&X\\
	\bv \cdot \nabla \delta \varphi(\bx, \bv) + \sigma_{x,tf} \delta \varphi(\bx, \bv) &= \sigma_{x,s}\cK \delta \varphi(\bx, \bv) + (\eta \sigma_{x,f} - \tilde{\eta}\tilde{\sigma}_{x,f})\cI \fW\qquad&\text{ in }&X\\
	\bv\cdot \nabla \delta{w}(\bx, \bv) + (\sigma_{m,s} + \sigma_{m,a})\delta{w}(\bx, \bv) &= \sigma_{m,s}\cK \delta{w}(\bx, \bv) + \eta \sigma_{x,f} \cI u - \tilde{\eta}\tilde{\sigma}_{x,f}\cI \tilde{u}\qquad&\text{ in }&X\\
	\delta\varphi(\bx, \bv) = 0,\quad &\delta w(\bx, \bv) = 0,	\quad \delta {u}(\bx, \bv) = 0\qquad&\text{ on }&\Gamma_{-}.
	\end{aligned}
	\end{equation}
	Then we can write the solutions $\delta u, \delta\varphi, \delta w$ as
	\begin{equation}
	\begin{aligned}
	\delta u &= -(I - \cL_x^{-1}\cS_x)^{-1} \cL_x^{-1}  (\delta \sigma_{x,f}\tilde{u})\\
	\delta \varphi &= (I - \cL_x^{-1}\cS_x)^{-1} \cL_x^{-1}  ((\delta \eta) \sigma_{x,f}\cI\fW + \tilde{\eta}\delta\sigma_{x,f}\cI\fW)\\
	\delta w &= (I - \cL_m^{-1}\cS_m)^{-1}  \cL_m^{-1} ((\delta\eta)\sigma_{x,f}\cI u + \tilde{\eta}(\sigma_{x,f}\cI u - \tilde{\sigma}_{x,f}\cI\tilde{u})),
	\end{aligned}
	\end{equation}
	where $\delta\eta = \eta - \tilde{\eta}$, $\delta \sigma_{x,f} = \sigma_{x,f} - \tilde{\sigma}_{x,f}$. We then decompose $(S-\tilde{S})$ into two parts: the first part is a Fredholm operator which acts on $\delta \eta$ and the second part is from the perturbation in $\sigma_{x,f}$,
	\begin{equation*}
	\begin{aligned}
	S(\bx) -\tilde{S}(\bx) &= -(\sigma_{m,a} + \sigma_{m,s})\intS \delta w(\bx, \bv)\fW(\bx, \bv) d\bv + \sigma_{m,s}\intS \cK\delta w(\bx, \bv)\fW(\bx, \bv) d\bv \\&\quad + (\delta \eta) \sigma_{x,f}(\cI u)(\cI \fW) +\eta (\delta\sigma_{x,f}) (\cI u)(\cI\fW) + \eta \sigma_{x,f}(\cI \delta u)(\cI \fW) \\&\quad - \delta \sigma_{x,f} \intS u(\bx, \bv) \varphi(\bx, \bv) d\bv - \sigma_{x,tf} \intS \delta u(\bx, \bv) \varphi(\bx, \bv) d\bv - \sigma_{x,tf} \intS u(\bx, \bv) \delta \varphi(\bx, \bv) \\&\quad  + \sigma_{x,s}\intS \cK\delta u(\bx, \bv) \varphi(\bx, \bv) d\bv + \sigma_{x,s}\intS \cK u(\bx, \bv) \delta \varphi(\bx, \bv)\\
	&= (\cA_1 + \cA_2 +\cA_3)\delta\eta + \cR.
	\end{aligned}
	\end{equation*}
	It is easy to verify that the reminder $\cR$ has a trivial bound 
	\begin{equation}
	\|\cR\|_{L^2(\Omega)} \le C \|\delta\sigma_{x,f}\|_{L^{\infty}(\Omega)}
	\end{equation}
	for some constant $C$. From the averaging lemma~\cite{devore2001averaging,diperna1991lp}, $\cA_1+\cA_2+\cA_3$ is Fredholm. Thus if $0$ is not an eigenvalue, we have the invertibility of $\cA_1+\cA_2+\cA_3$ and 
	\begin{equation}
	\|\delta \eta \|_{L^2(\Omega)} \le C'  \|S - \tilde{S}\|_{L^2(\Omega)} + C''\|\delta\sigma_{x,f}\|_{L^{\infty}(\Omega)}
	\end{equation}
	for some constants $C'$ and $C''$. 
\end{proof}
\section{Numerical experiments}\label{SEC:NUM}
The forward solver of RTE has been studied extensively in recent years and there are many existing numerical algorithms~\cite{ren2016fast,gao2009fast,larsen2002simplified,fiveland1988three}. In our work, we implement the forward solver by the discrete ordinate method with low order collocation scheme, where the phase space $X = \Omega\times\bbS^{d-1}$ is discretized in both spatial and angular space. In the physical space $\Omega$, we take the uniform mesh, on which the nodes are denoted by $\{\bx_i\}_{i=1}^N$. In the angular space $\bbS^{d-1}$, we uniformly choose the angular directions $\{\bv_k\}_{k=1}^M$ for each node $\bx_i$. For a medium with weak scattering, the solution is solved quickly by the following source iteration:
\begin{equation}\label{EQ:SI}
\begin{aligned}
u^{T+1}(\bx_i, \bv_k) &= g(\bx_i - \tau_{-}(\bx_i, \bv_k)\bv_k)\exp\left(- \int_{0}^{\tau_{-}(\bx_i, \bv_k)} \sigma_{x,tf}(\bx_i - s\bv_k) ds \right) \\&\quad + \int_{0}^{\tau_{-}(\bx_i, \bv_k)} \exp\left( -\int_0^l   \sigma_{x,tf}(\bx_i - s\bv_k) ds  \right) \sigma_{x,s}\cK u^T(\bx_i-l\bv_k, \bv_k) dl,
\end{aligned}
\end{equation}
where $u^T(\bx,\bv)$ denotes the solution at the $T$-th iteration. Along each direction $\bv_k$, the source iteration~\eqref{EQ:SI} can be computed with complete independence, hence the algorithm has a natural parallelism. Regarding the path integrals, we use the trapezoid rule for the first path integral term, which represents the ballistic contribution, and for the second path integral term, we use Simpson's rule. The paralleled forward solver is implemented in C++ and wrapped with MATLAB's mex interface, the source code is hosted at Github\footnote{\href{https://github.com/lowrank/rte}{https://github.com/lowrank/rte}}. 

Although the Theorem~\ref{THM:UNIQ NONLINEAR} implies a constructive way to get an approximated estimate of $\sigma_{x,f}$, the singular localized sources require very fine mesh to resolve, which is not practical for numerical simulation with the discrete ordinate method. However, we still can seek for the reconstruction of $\sigma_{x,f}$ by minimizing the following objective functional:
\begin{equation}
J[\sigma_{x,f}] =  \frac{1}{2} \int_{\Omega} |H - H^{\ast}|^2 d\bx + \frac{\beta}{2} \int_{\Omega} |\nabla \sigma_{x,f}|^2 d\bx,
\end{equation}
where $H^{\ast}$ is the synthetic internal data from the excitation stage and $\beta$ is the parameter of regularization. Using the linearization formula~\eqref{EQ:V}, we have
\begin{equation}
J'[\sigma_{x,f}](\delta \sigma_{x,f}) = \int_{\Omega} (H - H^{\ast}) \left[ -\delta \sigma_{x,f}\psi + \intS Q(\bx, \bv) v(\bx, \bv) d\bv \right] d\bx + \beta \int_{\Omega} \nabla \delta\sigma_{x,f}\cdot \nabla \sigma_{x,f} d\bx ,
\end{equation}
where $\psi = \intS u(\bx, -\bv) u(\bx, \bv) d\bv$ and $v$ is the solution to~\eqref{EQ:V}. The function $Q(\bx, \bv)$ is defined through
\begin{equation}
Q(\bx, \bv) = -2\sigma_{x,tf} u(\bx, -\bv) + 2\sigma_{x,s}\cK u(\bx, -\bv).
\end{equation}
We then use the quasi-Newton method (L-BFGS) to minimize the functional $J$. To simplify the evaluation process of the gradient, the adjoint state method is usually adopted. Let $q(\bx, \bv)$ be the solution to the adjoint RTE
\begin{equation}
\begin{aligned}
-\bv \cdot \nabla q(\bx, \bv) + \sigma_{x,tf} q(\bx, \bv) &= \sigma_{x,s}\cK q - (H - H^{\ast}) Q(\bx, \bv) \quad &\text{ in }&X,\\
q(\bx, \bv) &= 0\quad &\text{ on }&\Gamma_{+}.
\end{aligned}
\end{equation}
The gradient of $J$ is 
\begin{equation}
J'[\sigma_{x,f}](\delta \sigma_{x,f}) = \int_{\Omega} \delta\sigma_{x,f} \left[-(H - H^{\ast}) \psi + \intS q(\bx, \bv) u(\bx, \bv) d\bv \right] d\bx + \beta \int_{\Omega} \nabla \delta \sigma_{x,f} \cdot \nabla \sigma_{x,f} d\bx 
\end{equation}
The quantum efficiency $\eta$ is then recovered by solving the corresponding linear inverse source problem using the reconstructed $\sigma_{x,f}$, which is
\begin{equation}
\tilde{\eta} = \argmin_{\eta\in \cA_{\eta}} \frac{1}{2}\int_{\Omega} |S - S^{\ast}|^2 d\bx  + \frac{\beta'}{2} \int_{\Omega} |\eta|^2 d\bx,
\end{equation}
where $S^{\ast}$ is the computed internal data from the emission stage and  $\beta'$ is the Tikhonov regularization parameter in case the problem is ill-posed.

In the following numerical experiments, the physical domain is $\Omega = [0,1]^2 \subset \bbR^2$, and the scattering phase function is chosen as the Henyey-Greenstein's function $p_{HG}$ in two dimension,
\begin{equation}
p_{HG}(\cos\theta) = \frac{1}{2\pi}\frac{1 - \texttt{g}^2 }{1 + \texttt{g}^2 - 2\texttt{g} \cos \theta},
\end{equation} 
where the constant $\texttt{g}$ is the medium's anisotropy parameter. To avoid the inverse crime, for the following numerical experiments, the synthetic data are generated on a fine discretization on both physical and angular spaces, while the inverse problems are solved on a coarse discretized phase space with roughly 600,000 unknowns.
The numerical experiments are performed in MATLAB and the source code is hosted on Github~\footnote{\href{https://github.com/lowrank/fumot-rte/}{https://github.com/lowrank/fumot-rte/}.}.
\subsection{Example 1}\label{EX:1}
In this example, we demonstrate the nonuniqueness of the reconstruction of $\sigma_{x,f}$ in a medium with relatively strong scattering. Here $\sigma_{x,f}$ remains unknown on the entire domain $\Omega$. The coefficients are
\begin{equation}
\sigma_{x,s}(x, y) = 10 + 0.2 x,\quad \sigma_{x,a}(x,y) = 0.2 + 0.2 y,\quad\sigma_{x,f}(x,y) = 0.5 + 0.5 x.
\end{equation}
The illumination source is chosen as $g\equiv 1$ on the boundary and the anisotropy parameter $\texttt{g} = 0.5$, the initial guess of $\sigma_{x,f}$ is  set to zero. We also let the regularization parameter $\beta =0$ and assume noiseless internal data. In Fig~\ref{FIG:NONUNIQUE}, we can observe that the reconstructed image of $\sigma_{x,f}$ is completely different from the exact coefficient.
\begin{figure}[!htb]
	\centering
	\includegraphics[height=0.23\textwidth]{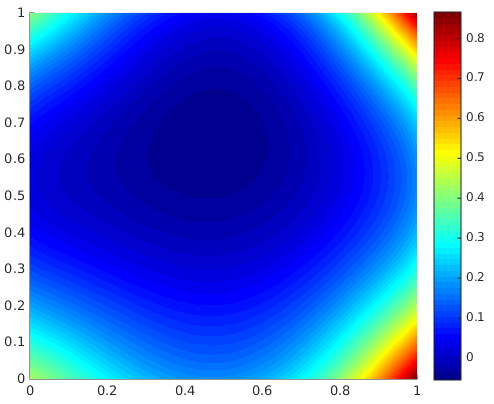}\qquad
	\includegraphics[height=0.23\textwidth]{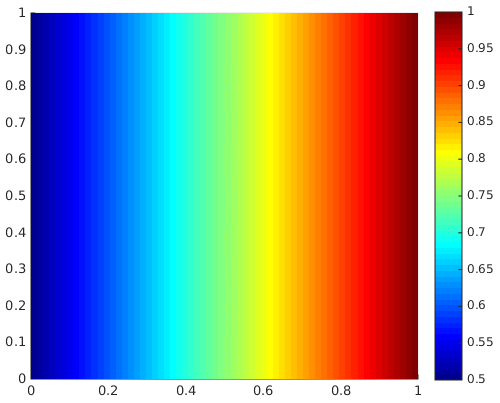}\quad 
	\includegraphics[height=0.23\textwidth]{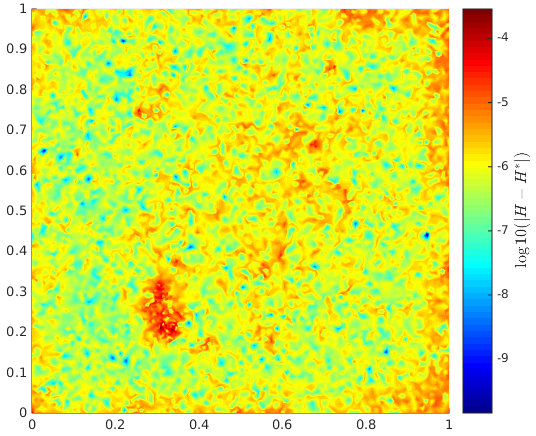}\quad 
	\caption{Nonuniqueness of $\sigma_{x,f}$. Left: The reconstructed $\sigma_{x,f}$. Middle: The exact $\sigma_{x,f}$. Right: The numerical difference between the exact internal data $H^{\ast}$ and the computed internal data $H$ from the reconstructed coefficient on the log scale. We can see that the difference between the internal data are quite small, however the difference between the coefficients is large.}\label{FIG:NONUNIQUE}
\end{figure}
\subsection{Example 2}
In this example, we consider an optically thin medium where $\sigma_{x,s}$ is moderately small and $\sigma_{x,f}$ remains unknown on the entire domain $\Omega$. We set coefficients
\begin{equation}
\begin{aligned}
\sigma_{x,s}(x, y) &= 0.2 + 0.2 x,\quad \sigma_{x,a}(x,y) = 0.2 + 0.2 y, \\
\sigma_{m,s}(x, y) &= 2.0 + 0.2 x,\quad \sigma_{m,a}(x,y) = 0.4 + 0.2 y,
\end{aligned}
\end{equation}
and let $\sigma_{x,f}$ be the modified Shepp-Logan phantom and $\eta$ the Derenzo phantom; see Fig~\ref{FIG:2XF}.
\begin{figure}[!htb]
	\centering
	\includegraphics[height=0.23\textwidth]{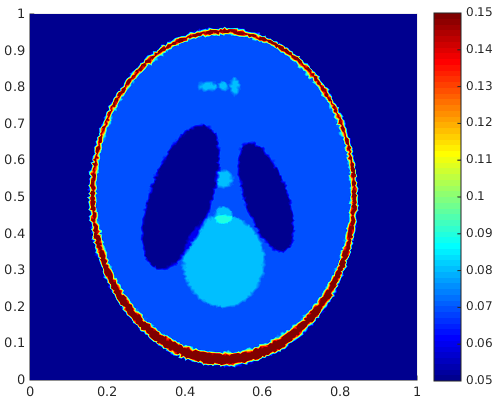}\qquad
	\includegraphics[height=0.23\textwidth]{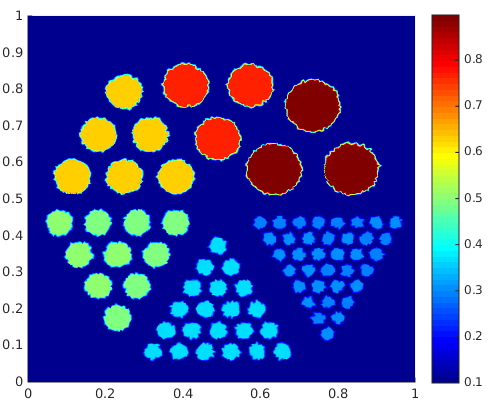}
	\caption{Left: the coefficient $\sigma_{x,f}$. Right: the quantum efficiency $\eta$.}\label{FIG:2XF}
\end{figure}
The anisotropy parameter $\texttt{g} = 0.5$ and the initial guess is generated randomly.  The illumination source $g$ is chosen as 
\begin{equation}
g(x, y) = 5\sin^2(4 \pi x) + 5\sin^2(4 \pi y),\quad (x,y)\in\partial\Omega.
\end{equation}
Such source simulates the ``singular" behavior in Theorem~\ref{THM:UNIQ NONLINEAR}, which results with relative strong signals along the lines between the ``points".  For the reconstruction, the regularization parameter is $\beta = 10^{-3}$, and the internal data is polluted by a multiplicative random noise ${H}^{\ast}\leftarrow H^{\ast}(1 +\tau \cU([-1,1]))$, with $\cU([-1,1])$ being the uniform distributed random variable and $\tau$ the noise level. The numerical reconstructions are shown in Fig~\ref{FIG:EX2}.
\begin{figure}[!htb]
	\centering
	\includegraphics[height=0.23\textwidth]{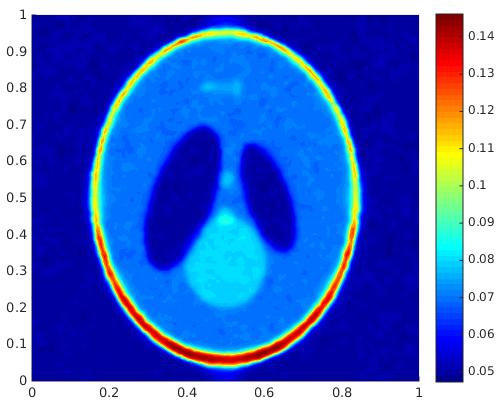}\qquad 
	\includegraphics[height=0.23\textwidth]{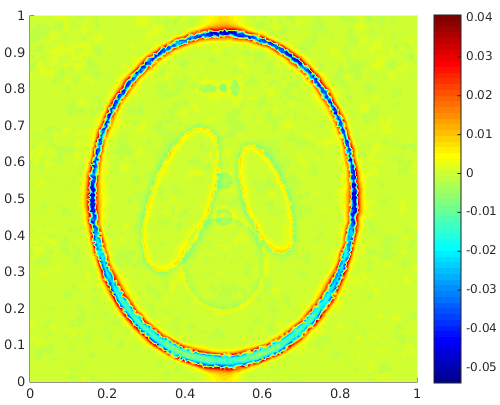}\\
	\includegraphics[height=0.23\textwidth]{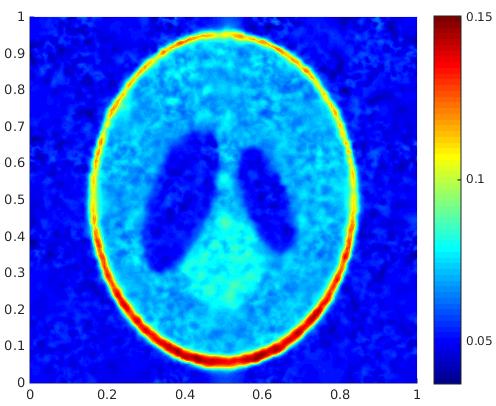}\qquad
	\includegraphics[height=0.23\textwidth]{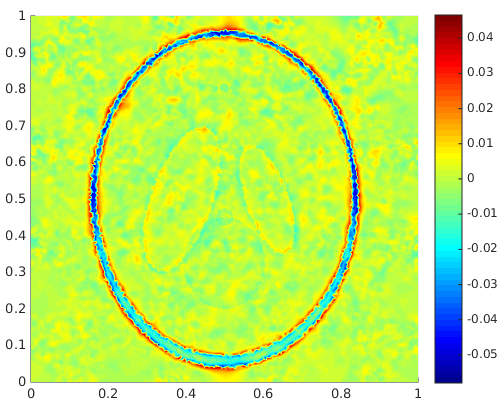}
	\caption{The reconstruction of $\sigma_{x,f}$. Top left: The reconstructed coefficient $\sigma_{x,f}$ with noise level $\tau = 1\%$. Top right: the error of $\sigma_{x,f}$, the relative $L^2$ error is $11.85\%$. Bottom left: The reconstructed coefficient $\sigma_{x,f}$ with noise level $\tau = 5\%$. Bottom right: the error of $\sigma_{x,f}$, the relative $L^2$ error is $12.67\%$. }\label{FIG:EX2}
\end{figure}
After the coefficient $\sigma_{x,f}$ has been recovered, we continue to use this $\sigma_{x,f}$ to reconstruct the quantum efficiency $\eta$ from the internal data $S^{\ast}$, we also pollute the data by a multiplicative random noise $S^{\ast} \leftarrow S^{\ast}(1 + \tau \cU([-1,1]))$ with the same noise level. The Tikhonov regularization parameter is $\beta' = 10^{-8}$. The corresponding numerical reconstructions are shown in Fig~\ref{FIG:EX2ETA}.
\begin{figure}[!htb]
	\centering
	\includegraphics[height=0.23\textwidth]{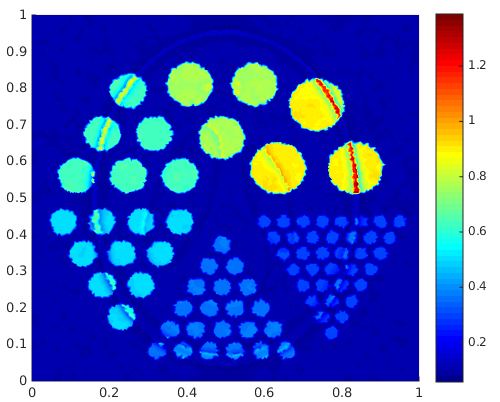}\qquad 
	\includegraphics[height=0.23\textwidth]{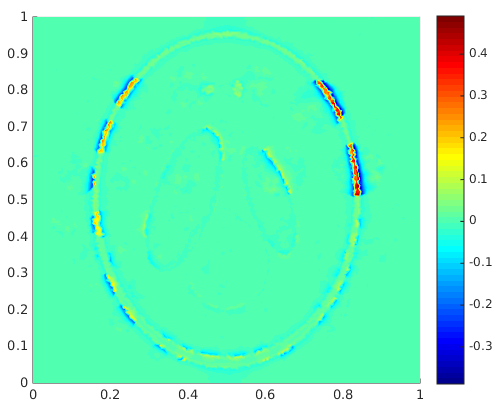}\\
	\includegraphics[height=0.23\textwidth]{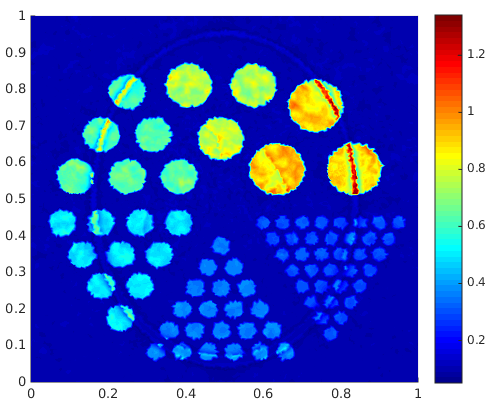}\qquad
	\includegraphics[height=0.23\textwidth]{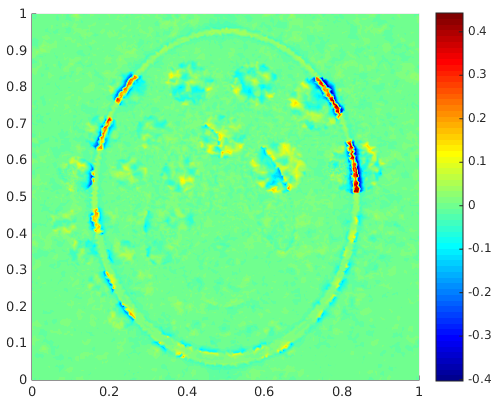}
	\caption{The reconstruction of $\eta$. Top left: The reconstructed coefficient $\eta$ with noise level $\tau = 1\%$. Top right: the error of $\eta$, the relative $L^2$ error is $9.61\%$. Bottom left: The reconstructed coefficient $\eta$ with noise level $\tau = 5\%$. Bottom right: the error of $\eta$, the relative $L^2$ error is $10.71\%$. }\label{FIG:EX2ETA}
\end{figure}

\section{Conclusion}\label{SEC:CONCLUSION}

In this paper, we studied the inverse problem in fluorescence ultrasound modulated optical tomography (fUMOT) in the transport regime with angularly averaged illumination and measurement. The inverse problem of interest is to recover the absorption coefficient of the fluorophores $\sigma_{x,f}$ and the quantum efficiency $\eta$.

We derived two internal functionals, $H(\bx)$ in \eqref{EQ:H} and $S(\bx)$ in \eqref{EQ:S}, from the boundary measurement. Assuming knowledge of the background optical coefficients $\sigma_{x,a}$, $\sigma_{m,a}$, $\sigma_{x,s}$ and $\sigma_{m,s}$, we investigated the uniqueness and stability of the nonlinear map $\sigma_{x,f}\mapsto H$ as well as its linearization $\delta\sigma_{x,f}\mapsto \mathfrak{H}$. For the linearized map, we showed $\delta\sigma_{x,f}$ is uniquely and stably determined by $\mathfrak{H}$ for optically thin media. For the nonlinear map, we proved $\sigma_{x,f}$ can be approximately reconstructed with properly chosen illumination sources, up to an error that can be made arbitrarily small. Upon successful recovery of $\sigma_{x,f}$, we proved the quantum efficiency $\eta$ is also uniquely and stably determined by the internal functional $S$; moreover, the error in the reconstruction of $\eta$ is controllable as long as that of $\sigma_{x,f}$ is. Finally, the resulting reconstruction procedures are numerically implemented to validate the theoretical conclusions.

\section*{Acknowledgment}

The research of YY was partly supported by the NSF Grant DMS-1715178, the AMS-Simons travel grant, and the start-up fund from the Michigan State University.

\bibliographystyle{siam}
\bibliography{main}
\end{document}